\documentclass[onecolumn,showpacs,amsmath,amssymb,preprint, superscriptaddress]{revtex4-1}
\usepackage{dcolumn}
\usepackage{bm}
\usepackage{hyperref}
\newcommand{\ignore}[1]{}
\usepackage{graphicx}
\usepackage{subfigure}
\usepackage{amsmath,amssymb,amsfonts,amsthm,color,latexsym}
\usepackage{soul}
\usepackage{cancel}

\usepackage[export]{adjustbox}
\usepackage{tabstackengine}

\def\mb{\mathbf}

\def\bm{\boldsymbol}
\newcommand{\Rmnum}[1]{\uppercase\expandafter{\romannumeral #1\relax}}
\newcommand{\rmnum}[1]{\lowercase\expandafter{\romannumeral #1\relax}}
\def\mb{\mathbf}
\newcommand*\diff{\mathop{}\!\mathrm{d}}

\mathchardef\mhyphen="2D
\usepackage{amsmath}
\usepackage{xcolor}

\newtheorem{theorem}{Theorem}[section]

\newtheorem{proposition}[theorem]{Proposition}

\usepackage[normalem]{ulem}

\begin{document}
\preprint{}

\title{Data-driven learning of the generalized Langevin equation with state-dependent memory}
\author{Pei Ge}
\affiliation{Department of Computational Mathematics, Science \& Engineering, Michigan State University, MI 48824, USA}%
\author{Zhongqiang Zhang}
\affiliation{Department of Mathematical Sciences, Worcester Polytechnic Institute, MA 01609, USA}%
\author{Huan Lei}
\email{leihuan@msu.edu}
\affiliation{Department of Computational Mathematics, Science \& Engineering, Michigan State University, MI 48824, USA}%
\affiliation{Department of Statistics \& Probability, Michigan State University, MI 48824, USA}%


\begin{abstract}
We present a data-driven method to learn stochastic reduced models of complex systems that retain a state-dependent memory beyond the standard generalized Langevin equation (GLE) with a homogeneous kernel. The constructed model naturally encodes the heterogeneous energy dissipation by jointly learning a set of state features and the non-Markovian coupling among the features. Numerical results demonstrate the limitation of the standard GLE and the essential role of the broadly overlooked state-dependency nature in predicting molecule kinetics related to conformation relaxation and transition.
\end{abstract}

\maketitle
\section{Introduction}
Predicting the collective behavior of complex multiscale systems is often centered around projecting the full-dimensional dynamics onto a set of resolved variables. However, an accurate construction of such a reduced model remains a practical challenge for real applications such as molecular modeling. While model reduction frameworks such as the Koopman operator \cite{Koopman315} and the Mori-Zwanzig projection formalism \cite{Mori1965, Zwanzig61} enable us to write down the dynamic equations in terms of the resolved variables, the reduced model generally becomes non-Markovian with a memory term that may further depend on the resolved variables; the direct numerical evaluation involves solving the expensive full-dimensional orthogonal dynamics. In practice, one common approximation is to ignore such state-dependency; the reduced model is simplified as the standard generalized Langevin equation (GLE) \cite{Zwanzigbook} with a memory kernel that only depends on time. Several approaches \cite{lange2006collective, Darve_PNAS_2009, ceriotti2009langevin, baczewski2013numerical, Dav_Voth_JCP_2015, Lei_Li_PNAS_2016,  russo2019deep, Jung_Hanke_JCTC_2017, Lee2019, ma2019coarse,wang2020data, Zhu2020, Klippenstein_Vegt_JCP_2021,  vroylandt2022likelihood, SheZ_JCP_2023, xie2023gle} have been developed to construct the memory kernel such that certain dynamic properties (e.g., the two-point correlations) can be properly reproduced. Despite its broad application, the validity of the standard GLE for real multiscale systems remains less understood \cite{Hanggi_Stochastic_book_1997,klippenstein2021introducing}.

Intuitively, the above model reduction problem is somewhat analogous to hiking on a mountain where the landscape map and the path roughness represent the free energy and the memory term, respectively. In general, 
we should not expect homogeneous path roughness
at the different locations (e.g., the valleys and the ridges), which, conversely, needs to be inferred from the hiking records. Indeed, studies based on full molecular dynamics (MD) simulations \cite{Posch_Balucani_Physica_A_1984, Straub_Berne_JPC_1987, Straub_Berne_JCP_1990, Plotkin_Wolynes_PRL_1998, Luo_Xie_JPCB_2006, Best_Hummer_PRL_2006, Best_Hummer_PNAS_2010, Hinczewski_Netz_JCP_2010, Satija_Makarov_JCP_2017, Morrone_Li_JPCB_2012, daldrop2017external} and sophisticated projection operator construction \cite{Deutch_Oppenheim_JCP_1971, Zwanzig73, Zwanzig_diffusion_JPC_1992,Berezhkovskii_Szabo_JCP_2011,Glatzel_Schilling_EPL_2021,Vroylandt_EPL_2022,Vroylandt_Monmarche_JCP_2022,ayaz2022generalized,Jung_Jung_JCP_2023} show that the extracted memory term can exhibit a pronounced state-dependent nature where the implications for the collective behaviors remain under-explored.
For extensive MD systems, a recent study \cite{Lyu_Lei_PRL_2023} on reduced modeling of polymer melt shows that the heterogeneous inter-molecular energy dissipation (i.e., the memory) can be crucial for transport on the hydrodynamic scale. However, for canonical non-extensive problems such as biomolecule kinetics, 
a quantitative understanding of the effect of state-dependent memory arising from intra-molecular interactions 
remains an open problem. 
Several recent works \cite{Lei_Li_PNAS_2016, Lee2019, Satija_Makarov_JPCB_2019, Grogan_Lei_JCP_2020, Singh_Mondal_JPCB_2021,  Ayaz_Netz_PNAS_2021, vroylandt2022likelihood, Dalton_Netz_PNAS_2023} investigate the non-Markovian implication for transition dynamics based on the standard GLE. 
While elegant semi-analytical studies \cite{Straub_Berne_JCP_1988, Singh_Krishnan_CPL_1990,Carmeli_Nitzan_CPL_1983,Tarjus_Kivelson_Chemical_Physics_1991,Krishnan_Singh_JCP_1992,Voth_JCP_1992,Straus_Voth_JCP_1993,Haynes_Voth_CPL_1993,Haynes_Voth_JCP_1994,  Cossio_Hummer_PNAS_2015} on idealized 1D double-well potential provide theoretical insights into the state-dependent nature, accurately quantifying such effect on molecule transition dynamics relies on accurate construction and efficient simulation of a reduced model beyond the standard GLE.

This work presents a data-driven approach for learning a stochastic reduced model that retains a state-dependent memory for non-extensive systems. Instead of dealing with the orthogonal dynamics \cite{Darve_PNAS_2009, Vroylandt_Monmarche_JCP_2022, Lyu_Lei_PRL_2023}, the training only relies on the trajectory samples
and the full model is unnecessarily known explicitly. 
The main idea is to seek a generalized representation of the memory as the composition of a set of state-dependent features, which encodes the coupling between the resolved and unresolved variables and can be learned using three-point correlation functions. Efficient training is achieved by constructing the encoder functions using a set of sparse bases, whose correlations can be efficiently pre-computed. Coherent noise can be introduced that ensures a consistent invariant distribution. The present method enables us to probe the open problem of the effect of state-dependent memory on molecular kinetics. Numerical results show that the broadly overlooked state-dependency can play a crucial role. In particular, the standard GLE could be insufficient to capture the collective properties such as the transition rate distribution, which, fortunately, can be reproduced by the present model. 

\section{Reduced model construction}
Let $(\mb q, \mb p) \in \mathbb{R}^{2m}$ represent the resolved variables of a high-dimensional Hamiltonian system, where $\mb q$ denotes the coarse-grained (CG) coordinates as a function of the position variables of the full model, and $\mb p$ denotes the CG momenta. Following the Zwanzig's projection formalism \cite{Zwanzigbook}, the reduced dynamics takes the form
\begin{equation}
\begin{split}
\dot{\mb q} &= \mb M^{-1} \mb p, \\
\dot{\mb p} &= -\nabla U(\mb q) - \int_0^t \mb K(\mb q(\tau), t - \tau) \mb v(\tau) \diff \tau +  \bm{\mathcal{R}}_t,
\end{split}
\label{eq:MZ_general}    
\end{equation}
where $\mb M$ is the mass matrix,  $U(\mb q)$ is the free energy, $\mb v := \dot{\mb q}$ is the velocity, $\mb K(\mb q, t)$ is the memory, and $\bm{\mathcal{R}}_t$ is the noise whose covariance function is related to the memory following the second fluctuation-dissipation theorem \cite{Vroylandt_Monmarche_JCP_2022}. Before proceeding to the construction of $\mb K(\mb q, t)$, we note that  $\mb M$ generally depends on $\mb q$. In this study, we focus on the effect of the state-dependent memory; the current choice of $\mb q$ leads to a constant mass matrix (see Refs. \cite{Lee2019,ayaz2022generalized} and Appendix for further discussion). Also, the construction of the free energy $U(\mb q)$ can be nontrivial; several canonical methods based on enhanced sampling \cite{Torrie_Valleau_Umbrella_JCP_1977,Kumar_Kollman_JCC_1992, Darve_Pohorille_JCP_2001, Laio_Parrinello_PNAS_2002} and temperature acceleration \cite{Tuckerman_adiabatic_JCP_2002,Eric_TAMD_CPL_2006, abrams2008efficient, Maragliano_Vanden_JCP_2008} have been developed. Here, we assume $U(\mb q)$ is known \emph{a priori}.

Instead of rigorously constructing $\mb K(\mb q, t)$ from the full model, we ask the question of which forms of $\mb K$ can generate a memory effect. 
One common approach is to embed the memory in a larger Markovian dynamics with a set of auxiliary variables. An essential observation is that the memory term can be generally written as
\begin{equation}
    \mb K(\mb q(\tau), t-\tau) \approx \mathcal{C}^+ \circ \exp\big( (t-\tau) \mathcal{L}_\text{aux} \big) \circ \mathcal{C}^-, 
\end{equation} 
where $\mathcal{L}_\text{aux}$ is the Liouville operator corresponding to the auxiliary dynamics and $\mathcal{C}^\pm$ are channels representing the coupling of the resolved and auxiliary variables. As a special case, if the coupling and the auxiliary dynamics take a linear form, the embedded memory recovers the standard GLE kernel, i.e., $\mb K(\mb q, t) = \mb K (t)$ (e.g., see Refs. \cite{Lei_Li_JCP_2021, SheZ_JCP_2023}). Therefore, to construct the reduced model beyond the standard GLE, the coupling channels need to properly retain certain kinds of state-dependency nature. This motivates us 
to represent $\mathcal{C}^\pm$ by seeking a set of state-dependent features $\phi(\mb q) = \left[\phi_1(\mb q), \cdots, \phi_n(\mb q)\right]$, where $\phi: \mathbb{R}^{m} \to \mathbb{R}^{n\times m}$ 
essentially encode the nonlinear coupling between the resolved and unresolved variables and the detailed form will be specified later. $\exp\left(t\mathcal{L}_\text{aux} \right)$ induces the non-Markovian interactions among the features with a time lag of $t$ characterized by a kernel function, i.e., $\mathcal{C}^{+} \exp\left((t-\tau)\mathcal{L}_\text{aux} \right) \circ \mathcal{C}^{-} =  \phi(\mb q(t))^T \Theta(t-\tau) \phi(\mb q(\tau))$, where $\Theta: \mathbb{R}^{+}\to \mathbb{R}^{n \times n}$ 
and component $\Theta_{ij}(t-\tau)$ represents the dissipation between features $\phi_i(\mb q(t))$ and $\phi_j(\mb q(\tau))$. In the remainder of this work, we use $\phi_t$ to denote $\phi(\mb q(t))$. 

Accordingly, reduced dynamics \eqref{eq:MZ_general} can be modeled by 
\begin{equation}
\begin{split}
\dot{\mb q} &= \mb M^{-1} \mb p, \\
\dot{\mb p} &= -\nabla U(\mb q) - \int_0^t \phi_t^T \Theta(t-\tau) \phi_\tau \mb v(\tau) \diff \tau +  \bm{\mathcal{R}}_t,
\end{split}
\label{eq:MZ_ansatz}    
\end{equation}
where encoders $\left\{\phi_i(\mb q)\right\}_{i=1}^n$ and kernel $\Theta(t)$ need to be determined. As a special case, at the Markovian limit $\Theta(t) \propto \delta(t)$, Eq. \eqref{eq:MZ_ansatz} recovers the Langevin dynamics and the quadratic form $\phi^T\phi$ ensures positive energy dissipation. Also, by choosing $\Theta(t)$ to be diagonal with individual components corresponding to certain frequency modes, Eq. \eqref{eq:MZ_ansatz} reduces to the heat bath model \cite{Zwanzig73} with a nonlinear coupling of bath coordinates. On the other hand, the present model enables an adaptive choice of the number of spatial 
features and a more general form of $\Theta(t)$ with the off-diagonal components capturing the non-Markovian coupling among the features, which turns out to be crucial for reproducing the collective dynamics.

We emphasize that 
Eq. \eqref{eq:MZ_ansatz} should not be viewed as a direct approximation of Zwanzig's projection formalism. Rather, it serves as a reduced model that faithfully retains the state-dependent memory effect. To construct the model, we represent encoders $\left\{\phi_i(\mb q)\right\}_{i=1}^n$ and kernel $\Theta(t)$ in form of
\begin{equation}
\begin{split}
\phi_i(\mb q) &=  \mb H_i^T \psi (\mb q), \\
\Theta(t) &= {\rm e}^{-\alpha t} \sum_{k=0}^{N_\omega} \hat{\Theta}_k \cos(\omega_k t),
\end{split}
\label{eq:phi_theta}
\end{equation}
where $\psi(\mb q) = \left[\psi_1(\mb q), \cdots, \psi_{N_b}(\mb q)\right]$ is a set of sparse bases and $\mb H = \left[\mb H_1^T, \cdots, \mb H_n^T\right]$ are trainable coefficients to be determined. In this work, we choose the piece-wise linear bases such that the correlation between $\phi_i$ and $\phi_j$ can be efficiently evaluated; 
other choices such as localized kernel can be also used. $\Theta(t)$ needs to preserve positive semi-definiteness. Hence, we represent Fourier modes $\hat{\Theta}_k = \Gamma_k\Gamma_k^T$, where $\Gamma_k \in \mathbb{R}^{n\times n}$ is a low-triangular matrix to be determined along with $\alpha \ge 0$. For the present study, the full dynamics is reversible; $\alpha$ approaches $0$ and essentially serves as a regularization parameter. We note that $\Theta(t)$ in the form of Eq. \eqref{eq:phi_theta} can be further generalized by introducing an anti-symmetry part and refer to Appendix for further discussion.  

To learn the reduced model \eqref{eq:MZ_ansatz}, we need to choose appropriate metrics such that the state-dependent non-Markovian nature can be manifested. While auto-correlation functions such as $c_{vv}(t) = \left\langle \mb v(t)\mb v(0)^T\right\rangle$ merely characterize the overall memory effect, a crucial observation is that the correlations conditional with different initial state $\mb q_0$ further depends on the local energy dissipation and therefore naturally encodes the signatures of the heterogeneous memory effect. Accordingly, we right-multiply the second equation of Eq. \eqref{eq:MZ_ansatz} by $\mb v(0)$ and take the conditional expectation on $\mb q_0 = \mb q^{\ast}$, i.e., 
\begin{equation}
\begin{split}
\mb g(t; \mb q^{\ast}) &= \int_0^t \left\langle\phi_t^T \Theta_{t-\tau} \phi_{\tau} \mb v_\tau \mb v_0^T\vert \mb q_0 = \mb q^{\ast}\right\rangle \diff \tau  \\
&= \int_0^t  \left\langle {\rm Tr} \left[\Theta_{t-\tau} \mb H \psi_\tau \mb v_\tau \mb v_0^T \psi_t^T \mb H^T \right] \vert \mb q_0 = \mb q^{\ast} \right\rangle \diff \tau \\    
&= \int_0^t {\rm Tr} \left[ \Theta_{t-\tau} \mb H \mb C_{\psi,\psi}(t, \tau; \mb q^{\ast}) \mb H^T \right] \diff \tau,
\end{split}
\nonumber
\end{equation}
where $\mb g(t; \mb q^{\ast}) := \left\langle [\dot{\mb p}_t + \nabla U(\mb q_t)] \mb v_0^T\vert \mb q_0 = \mb q^{\ast}\right\rangle$ and $\mb C_{\psi,\psi}(t, \tau; \mb q^{\ast}) := \left\langle \psi_\tau \mb v_\tau \mb v_0^T \psi_t^T \vert \mb q_0 = \mb q^{\ast}\right\rangle$ 
is a three-point correlation characterizing the coupling among the bases. Since $\psi(\mb q)$ is sparse, $\psi_\tau \psi_t^T$ can be evaluated with $O(1)$ complexity and hence $\mb C_{\psi,\psi}(t, \tau; \mb q^{\ast})$ can be efficiently pre-computed. Accordingly, we can establish the training of the reduced model in terms of coefficients $\mb H$ for encoders $\phi(\mb q)$ as well as matrices $\left\{\Gamma_k\right\}_{k=1}^{N_\omega}$ and $\alpha$ for kernel $\Theta(t)$ by minimizing the empirical loss
\begin{equation}
\begin{split}
L &= \sum_{l=1}^{N_q}\sum_{k=1}^{N_t} \left\Vert  \widetilde{\mb g} (t_k; \mb q^{(l)}) - \mb g(t_k; \mb q^{(l)})  \right\Vert^2, \\
\widetilde{\mb g}(t_k; \mb q^{(l)}) &= \sum_{j=1}^k {\rm Tr} \left[\Theta(t_k-t_j) \mb H \mb C_{\psi,\psi}(t_k, t_j; \mb q^{(l)}) \mb H^T \right] \delta t,
\end{split}
\nonumber
\end{equation}
where $l$ and $k$ correspond to various initial states and discrete time, respectively. $\delta t$ is the time step of training samples. $\widetilde{\mb g}(\cdot)$ represents the prediction by the reduced model which depends on the trainable variables and the pre-computed correlation $\mb C_{\psi,\psi}$.

To simulate the reduced model \eqref{eq:MZ_ansatz} \eqref{eq:phi_theta} on $t\in[0, ~T]$, we generate coherent noise $\bm{\mathcal{R}}_t = \phi_t^T \widetilde{\mb R}(t)$, where $\widetilde{\mb R}:\mathbb{R}^{+}\to\mathbb{R}^{n}$ is a Gaussian random process. Specifically, we can show that by choosing $\langle \widetilde{\mb R}(t) \widetilde{\mb R}(\tau)^T\rangle = k_BT{\rm e}^{-\alpha (t-\tau)}\Theta(t-\tau)$,  the reduced model retains a consistent equilibrium density, i.e., $\rho_{\rm eq}(\mb q, \mb p) \propto \exp\left\{-\beta \left[U(\mb q) + \frac{1}{2} \mb p^T \mb M^{-1}\mb p\right]\right\}$ (see proof in Appendix). Accordingly, we can generate discrete samples $\{\widetilde{\mb R}(t_i)\}_{i=0}^{N}$ 
by
\begin{equation}
\widetilde{\mb R}(t_i) = \beta^{-1/2} \sum_{k=0}^{2N} \widetilde{\Theta}_k^{1/2} \left[\cos(\omega_k t_i) \xi_k + \sin(\omega_k t_i) \eta_k\right],
\label{eq:random_noise}
\end{equation}
where $\widetilde{\Theta}_k$ are the Fourier (essentially cosine) modes of ${\rm e}^{-\alpha \vert t\vert }\Theta(t)$ on $\left[-T, T\right]$ (see Refs. \cite{berkowitz1983generalized, ogorodnikov1996numerical} and Appendix for the analytical form); $\xi_k$ and $\eta_k$ are independent Gaussian random vectors. 
In practice, $\widetilde{\mb R}(t)$ by Eq. \eqref{eq:random_noise} can be generated using FFT \cite{Coole_Tukey_FFT_1965} and the convolution term $\int_0^t \phi_t^T \Theta(t-\tau) \phi_\tau \mb v(\tau) \diff \tau$ in Eq. \eqref{eq:MZ_ansatz} can be efficiently evaluated using the fast convolution algorithm \cite{Achim_SIAM_2006}, both of which only require $O(N\log N)$ complexity. 

\section{Numerical results}
The present reduced model enables us to systematically investigate the open problem of the effect of state-dependent memory on the collective dynamics of complex systems such as molecule kinetics. 
In this work, we
consider the molecule benzyl bromide in an aqueous environment. The full MD system consists of one benzyl bromide molecule and 2400 water molecules with the periodic boundary condition. The isothermal-isobaric thermostat \cite{Martyna_Klein_JCP_1994} is used to equilibrate the system at 298K and 1 bar and a canonical ensemble with a Nos\'e-Hoover thermostat \cite{Nose_Mol_Phys_1984, Hoover1985} is used for the production stage. The resolved variable $\mb q$ characterizes the interplay between the substituent and the benzene group and is defined as the distance between the bromine atom and the ipso-carbon atom. 

\begin{figure}[htbp]
    \centering
    \includegraphics[width=0.95\linewidth]{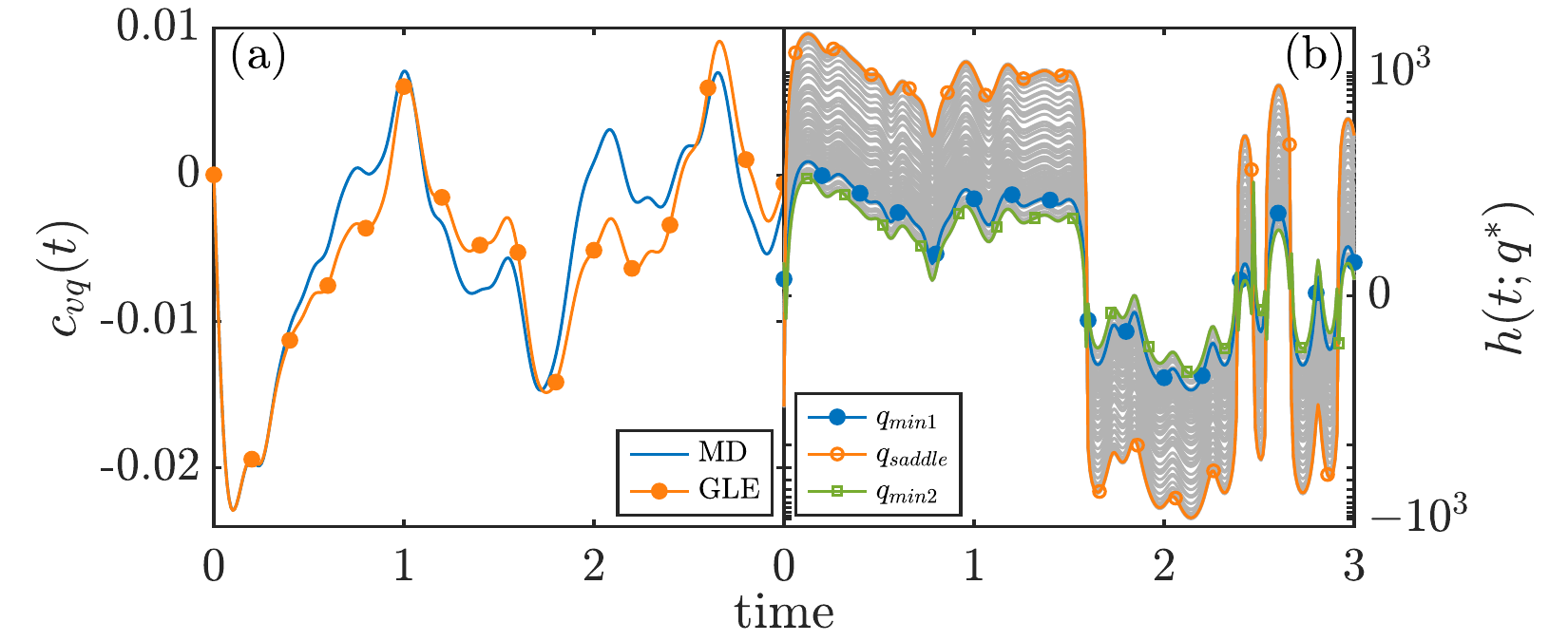}
    \caption{Correlation functions predicted by the standard GLE and the full MD: (a) Overall $c_{vq}(t)$ and (b) conditional 
 $h(t;\mb q^{\ast})$ with $\mb q^{\ast}$ representing various initial states (gray lines), including two local minima and the saddle point (see inset of Fig. \ref{fig:phi_theta}(a)). The large dispersion implies the limitation of the standard GLE, which predicts a single curve for $h(t)$ in short time.} 
\label{fig:c_vq_h}
\end{figure}

Let us start with the standard GLE by setting features $\phi(\mb q) \equiv \mb I$ in Eq. \eqref{eq:MZ_ansatz}. We right-multiply $\mb q(0)$ (or $\mb v(0)$) to Eq. \eqref{eq:MZ_ansatz} and compute the correlation functions, i.e., $h(t) = \int_0^t \Theta(t-\tau) c_{vq}(\tau) \diff \tau$, where $h(t) = \left\langle [\dot{\mb p}_t + \nabla U(\mb q_t)] \mb q_0^T\right\rangle$. The standard GLE kernel $\Theta(t)$ (i.e., $\mb K(t)$ in Eq. \eqref{eq:MZ_general}) can be computed using the Fourier transform of the integral equation.
If the reduced dynamics \eqref{eq:MZ_general} can be simplified as the standard GLE, then $c_{vq}(t)$ should be accurately reproduced.   
Fig. \ref{fig:c_vq_h} shows the prediction of $c_{vq}(t)$ from the standard GLE and the full MD model. The apparent deviations imply non-negligible state-dependency. To further probe this effect, we compute $h(t; \mb q^{\ast})$ conditional with different initial states $\mb q^{\ast}$. Unlike a unified correlation predicted by the standard GLE, the large dispersion reveals the heterogeneous nature of the energy dissipation process.

\begin{figure}[htbp]
    \centering
    \includegraphics[width=0.95\linewidth]{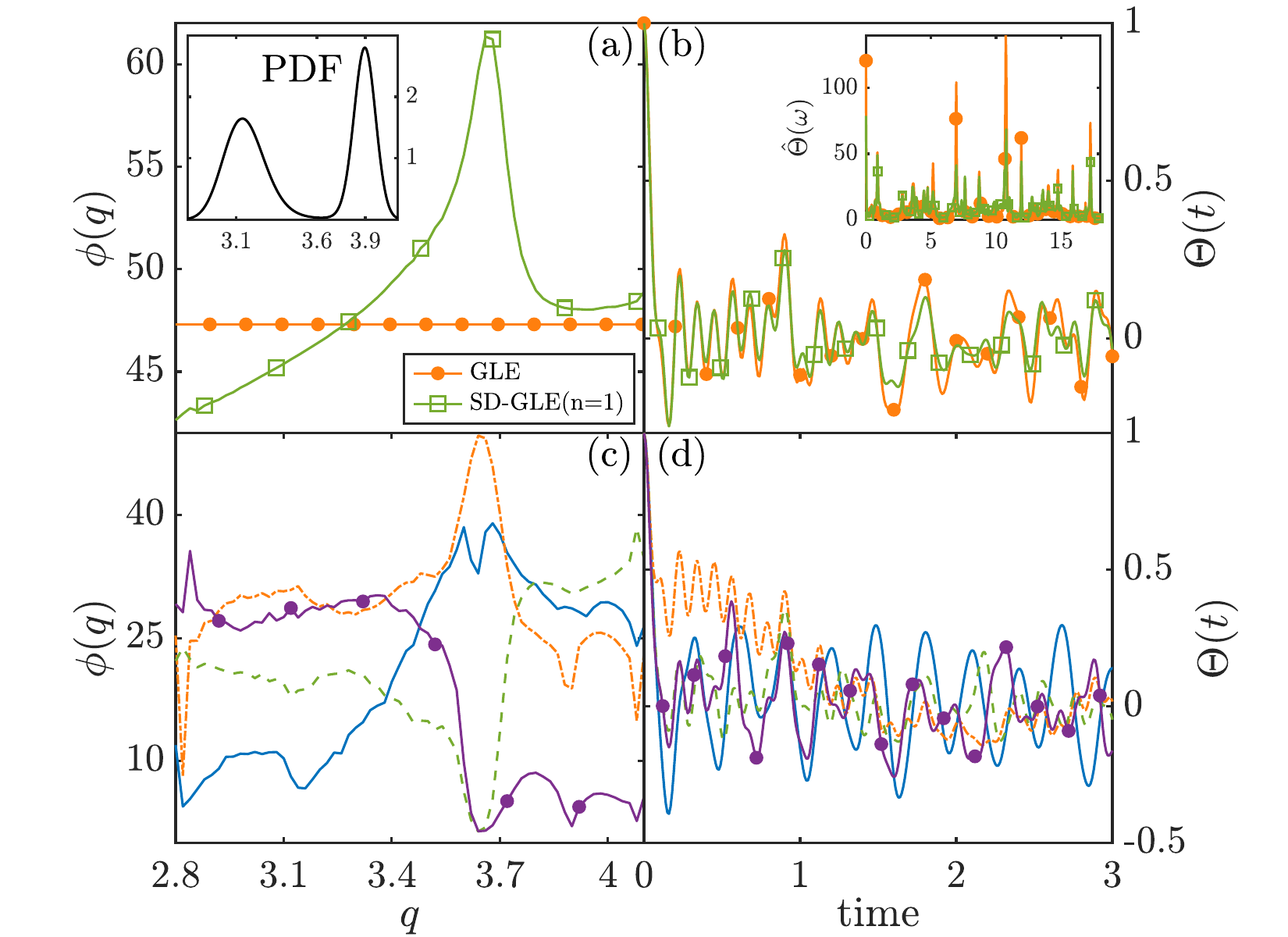}
    \caption{The state features $\phi$ and diagonal components of the matrix-valued kernel $\Theta(t)$ for the present model with state-dependent memory (SD-GLE) trained using (a-b) one feature and (c-d) four features. Inset plots: (a) probability density function (PDF) of $q$, where $\phi(q)$ near the saddle point shows a pronounced peak; (b) Fourier modes of $\Theta(t)$ for the present model and standard GLE. }
    \label{fig:phi_theta}
\end{figure}

To capture the state-dependent memory, we train the present model \eqref{eq:MZ_ansatz} with a different number of features. Fig. \ref{fig:phi_theta}(a-b) shows the obtained encoder $\phi(\cdot)$ using one feature and $\Theta(t)$ is scaled such that $\Theta(0) = 1$. We can see that $\phi$ exhibits apparent deviation from a uniform distribution. In particular, it shows a peak value near the saddle point $q = 3.65$, implying a larger effective friction near the regime. This result supports a similar assumption in earlier semi-analytical studies (e.g., see Ref. \cite{Straus_Voth_JCP_1993}) on improving Kramers' rate theory \cite{Kramers_1940}. 
Also, it explains the short-time dispersion shown in Fig. \ref{fig:c_vq_h}, where $h(t; \mb q^{\ast})$ with $\mb q^{\ast}$ near the saddle point decays faster than the ones near the local minima.  Furthermore, we can train the reduced model using multiple features. Fig. \ref{fig:phi_theta}(c-d) shows the obtained encoders $\{\phi_i(\cdot)\}_{i=1}^n$ with $n=4$ and the diagonal components of $\Theta(t)$. Compared with the case of $n=1$, the larger variation of $\phi_i$ enables a better representation of the state-dependent memory.

\begin{figure}[htbp]
    \centering
    \includegraphics[width=0.95\linewidth]{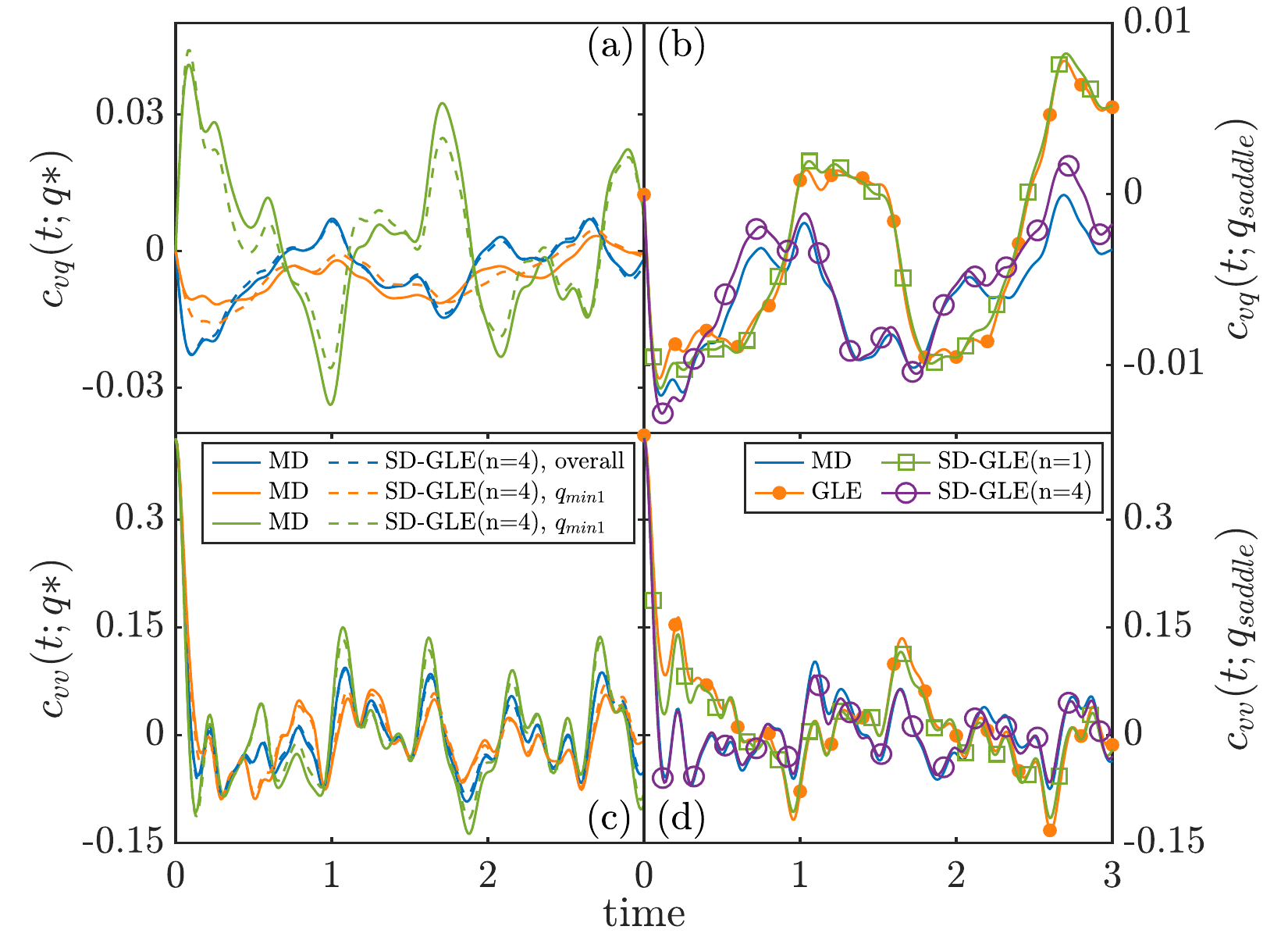}
    \caption{The overall and conditional correlation functions predicted by the full MD and various reduced models for two local minimal and the saddle point: (a-b) $c_{vq}$ and (c-d) $c_{vv}$. }
    \label{fig:c_vq_c_vv_MD_reduced_model}
\end{figure}

With the constructed model, we examine the conditional correlations $c_{vq}(t; \mb q^{\ast})$ and $c_{vv}(t; \mb q^{\ast})$. As shown in Fig. \ref{fig:c_vq_c_vv_MD_reduced_model}, for both the local minima and the saddle point, the predictions of the present model using four features show good agreement with the full MD results. In contrast, the 
predictions of the standard GLE show apparent deviations for $\mb q^{\ast}$ as the saddle point. Also, we note that the present model using one feature shows improved short-time predictions but remains insufficient for long-time correlations. This reveals the complex global variation of the memory term,
which can not be simply represented by a single feature as a state-dependent re-scaling of the kernel function;
the non-Markovian coupling among multiple features is crucial to capture the heterogeneous energy dissipation over the full space.

Finally, we examine the collective behavior related to molecule kinetics. Fig. \ref{fig:c_qq_transition}(a) shows the position correlation $c_{qq}(t)$ characterizing the relaxation of the molecule conformation. Compared with the MD results, the standard GLE shows a significant underestimation of the overall relaxation time. This discrepancy is possibly due to the larger effective friction near the saddle point (see Fig. \ref{fig:phi_theta}(a)), which essentially dampens the transition between the two local minima. The standard GLE overlooked such state dependency and therefore yields a faster relaxation. This limitation is consistently reflected in the distribution of the transition time. As shown in Fig. \ref{fig:c_qq_transition}(b), the standard GLE predicts a larger probability for the short transition time, indicating a smaller overall friction than the local (i.e., saddle point) value.  Fortunately, the heterogeneous non-Markovianity can be faithfully retained in the present model. In particular, the constructed model using a single feature yields a better prediction than the standard GLE. As we increase to four features, the predictions recover the MD results. 

\begin{figure}[htbp]
    \centering
    \includegraphics[width=0.95\linewidth]{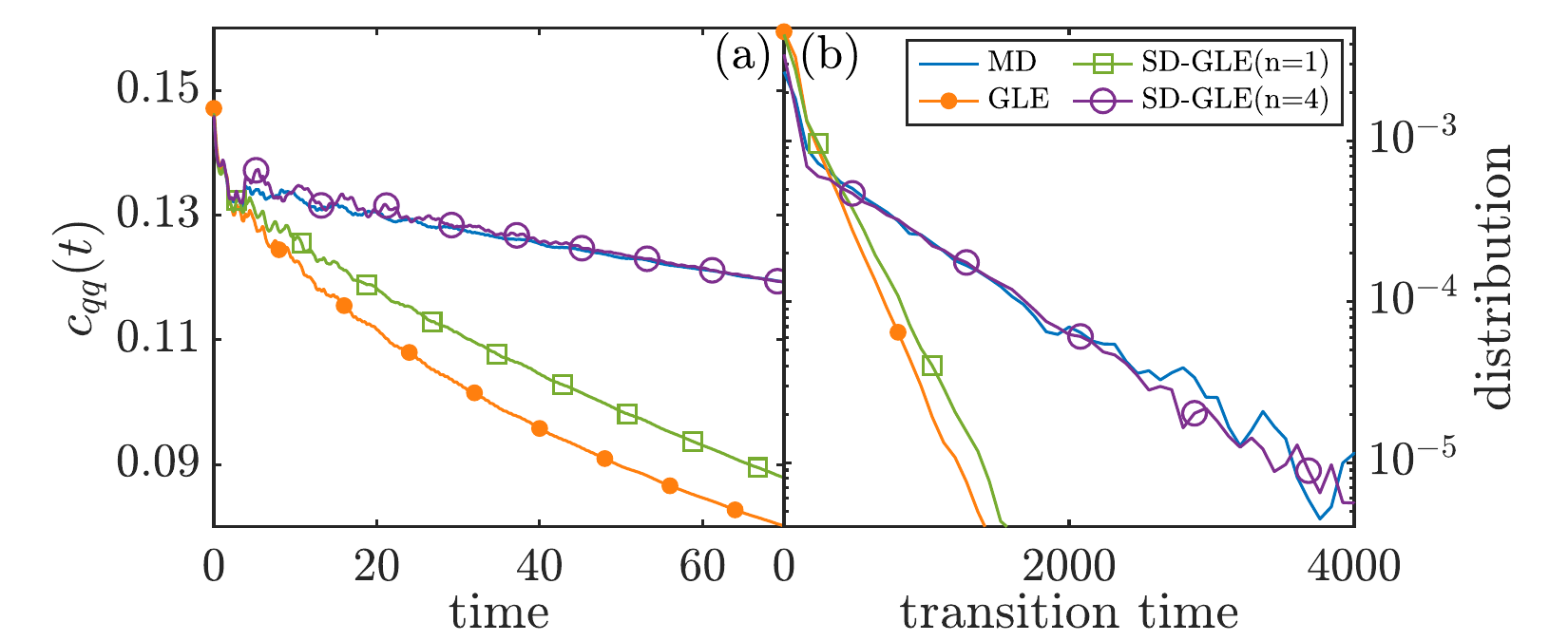}
    \caption{Collective molecule behaviors predicted by the full MD and the various reduced models: (a) overall conformation relaxation and (b) distribution of the transition time between the two local minima.}
    \label{fig:c_qq_transition}
\end{figure}

\section{Summary} 
%
To conclude, we present a data-driven approach for learning a stochastic reduced model beyond the standard GLE, where the complex state-dependent memory can be naturally encoded in the non-Markovian interactions among a set of features in terms of the resolved variables. 
The training does not rely on the explicit knowledge of the full model and only utilizes the trajectory samples, where the three-point correlations can be efficiently pre-computed.  
Numerical results of a molecule system demonstrate the crucial role of the state-dependent non-Markovianity on collective behavior, where the standard GLE shows limitations due to the over-simplified assumption of a homogeneous memory kernel.  In contrast, the present model accurately predicts the molecule kinetics including the transition time distribution, and paves the way toward predictive modeling of the collective functional properties and rare events \cite{E_Vanden_ARPC_2010} of complex multi-scale biomolecule and material systems.

\begin{acknowledgements}
We acknowledge helpful discussions from Zhaopeng Hao and Peiyuan Gao. The work is supported in part by the National Science Foundation under Grant DMS-2110981 and the ACCESS program through allocation MTH210005. 
\end{acknowledgements}

\appendix
\section{Mass matrix of the reduced model}
In this study, we focus on the effect of the state-dependent non-Markovian memory on the collective behavior of complex systems. We choose the coarse-grained resolved variables such that the corresponding mass matrix is a constant. On the other hand, the mass matrix should further depend on the resolved variables for the general cases, and we refer to Refs. \cite{ayaz2022generalized, Lee2019} for further discussions and the reduced dynamics with position-dependent mass. 
Specifically, we define $q = \Vert \mb Q_1 - \mb Q_2\Vert$, where $\mb Q_1$ and $\mb Q_2$ are the atom coordinates of the full model (see Fig. \ref{fig:mol_blb} and Sec. \ref{sec_SI:MD} for details). Accordingly, we have $\dot{q} = \mb Q_{12}^T \dot{\mb Q}_{12}/q$ and its covariance follows 
\begin{equation}
\begin{split}
\left\langle \dot{q} \dot{q} \right\rangle  &= \left\langle \frac{1}{q^2} \mb Q_{12}^T \dot{\mb Q}_{12} \dot{\mb Q}_{12}^T \mb Q_{12} \right\rangle  \\
&= \left\langle \frac{1}{q^2} {\rm Tr}\left[ (\mb Q_{12} \mb Q_{12}^T) (\dot{\mb Q}_{12} \dot{\mb Q}_{12}^T)\right] \right\rangle \\
&= \left\langle \frac{1}{q^2}  {\rm Tr}(\mb Q_{12} \mb Q_{12}^T) \right\rangle \left(M_1^{-1} + M_2^{-1}\right)k_BT  \\
&= \left(M_1^{-1} + M_2^{-1}\right)k_BT,
\end{split}
\end{equation}
where $M_1$ and $M_2$ represent the mass of two atoms and we have used the fact that the distribution of $\mb Q_{12}$ and $\dot{\mb Q}_{12}$ are independent. Therefore, the mass matrix of $q$ is a constant $M \equiv M_1M_2/(M_1+M_2)$.

\section{Coherent noise and invariant density of the reduced model}
We construct the reduced model
\begin{equation}
\begin{split}
\dot{\mb q} &= \mb M^{-1} \mb p, \\
\dot{\mb p} &= -\nabla U(\mb q) - \int_0^t \phi_t^T \Theta(t-\tau) \phi_\tau \mb v(\tau) \diff \tau +  \bm{\mathcal{R}}_t,
\end{split}
\label{eq_SI:MZ_ansatz}    
\end{equation}
where $\phi_t = \left[\phi_1(\mb q_t), \cdots, \phi_n(\mb q_t)\right]$ and $\Theta(t)$ represent the spatial features and the kernel to be learned. In particular, $\Theta(t)$ is directly constructed in the Fourier space, i.e., $\Theta(t) = {\rm e}^{-\alpha t} \sum_{k=0}^{N_\omega} \hat{\Theta}_k \cos(\omega_k t)$, where $\omega_k = \frac{2\pi}{T_c}k$ and $T_c$ is the time domain cut-off of the kernel. We note that ${\rm e}^{-\alpha t}$ should not be viewed as the bases to approximate $\Theta(t)$ (e.g., $\left\{{\rm e}^{-\alpha_i t} \cos (\beta_i t), {\rm e}^{-\alpha_i t} \sin (\beta_i t)\right\}_{i=1}^{N_\alpha}$; see Refs. \cite{Lei_Li_PNAS_2016, Lee2019}). Rather, $\Theta(t)$ is mainly characterized by the Fourier series expansion on $[0, T]$, and the exponential term ${\rm e}^{-\alpha t}$ is essentially a regularization term to eliminate the periodicity while maintaining the semi-positive definiteness condition.

For the fluctuation term $\bm{\mathcal{R}}_t$, we represent it as a noise in the form of 
$\bm{\mathcal{R}}_t = \phi_t^T \widetilde{\mb R}(t)$, where $\widetilde{\mb R}(t)$ is a Gaussian random process whose covariance function determined by $\Theta(t)$, i.e., $\langle \widetilde{\mb R}(t) \widetilde{\mb R}(\tau)^T\rangle = k_BT{\rm e}^{-\alpha (t-\tau)}\Theta(t-\tau)$. This choice avoids dealing with the orthogonal dynamics to calculate the fluctuation term. Furthermore, we can show that this choice enables the reduced model to retain a consistent invariant density function. 

\begin{proposition}
\label{prop:markovian_form}
For reduced model \eqref{eq_SI:MZ_ansatz} with $\Theta(t) = {\rm e}^{-\alpha t} \sum_{k=0}^{N_\omega} \hat{\Theta}_k \cos(\omega_k t)$, by choosing the fluctuation term $\bm{\mathcal{R}}_t = \phi_t^T \widetilde{\mb R}(t)$, where $\widetilde{\mb R}(t)$ is a Gaussian random process satisfying 
\begin{equation}
\langle \widetilde{\mb R}(t) \widetilde{\mb R}(\tau)^T\rangle = k_BT{\rm e}^{-\alpha (t-\tau)}\Theta(t-\tau),     
\end{equation}
the reduced model has an invariant distribution 
\begin{equation}
\rho_{\rm eq}(\mb q, \mb p) \propto \exp\left\{-\left[U(\mb q) + \mb p^T \mb M^{-1}\mb p/2\right]/k_BT\right\}.     
\end{equation}
\end{proposition}

\begin{proof}
Let us introduce auxiliary variables
\begin{equation}
\begin{split}
\mb z_{k,1} &= -\int_0^t {\rm e}^{-\alpha (t-\tau)} \Gamma_k \cos(\omega_k (t-\tau)) \phi_\tau \mb v(\tau) \diff \tau + \mb R_{k,1}(t),\\
\mb z_{k,2} &=-\int_0^t {\rm e}^{-\alpha (t-\tau)} \Gamma_k  \sin(\omega_k (t-\tau)) \phi_\tau \mb v(\tau) \diff \tau + \mb R_{k,2}(t),
\end{split}
\label{eq_SI:z_12}
\end{equation}
where $\Gamma_k^T \Gamma_k = \hat{\Theta}_k$  and $\mb R_{k,1}(t)$ is a Gaussian random process satisfying 
\begin{equation}
\left\langle \mb R_{j,1}(t) \mb R_{k,1}(\tau)^T \right\rangle = k_BT \delta_{jk} {\rm e}^{-\alpha (t-\tau)} \cos(\omega_k (t-\tau)),
\label{eq_SI:covariance_R_1}
\end{equation}
where $\delta_{jk}$ is the Kronecker delta. Accordingly, the second equation of Eq. \eqref{eq_SI:MZ_ansatz} can be written as
\begin{equation}
\dot{\mb p} = -\nabla U(\mb q) + \phi(\mb q)^T \sum_{k} \Gamma_k^T \mb z_{k,1},
\label{eq_SI:p_dynamics}
\end{equation}
and $\mb R_{j,2}(t)$ will be specified later.

Let $\mb z_k = \left[\mb z_{k,1}, \mb z_{k,2}\right]$ and  $\mb R_k = \left[\mb R_{k,1}, \mb R_{k,2}\right]$, we can rewrite Eq. \eqref{eq_SI:z_12} by
\begin{equation}
\begin{split}
\mb z_k &=   -\int_0^{t} {\rm e}^{-\alpha (t-\tau)} \begin{pmatrix}
\cos(\omega_k (t-\tau))I &    \sin(\omega_k (t-\tau))I \\
-\sin(\omega_k (t-\tau))I & \cos(\omega_k (t-\tau))I
\end{pmatrix} \begin{pmatrix}
\Gamma_k \phi_{\tau}\mb v(\tau) \\
0
\end{pmatrix} \diff \tau + \mb R_k(t) \\
&= -\int_0^{t} \exp \left[\begin{pmatrix}
-\alpha  I & \omega_k  I \\
-\omega_k  I & -\alpha I
\end{pmatrix}
(t - \tau)\right] \begin{pmatrix}
\Gamma_k \phi_{\tau}\mb v(\tau) \\
0
\end{pmatrix} \diff \tau + \mb R_k(t). 
\end{split}
\label{eq_SI:z_integral}
\end{equation}
By taking the time derivative of Eq. \eqref{eq_SI:z_integral} with respect to $t$, we have 
\begin{equation}
\frac{\diff \mb z_k}{\diff t} = \underbrace{\begin{pmatrix}
-\alpha I & \omega_k  I \\
-\omega_k  I & -\alpha I
\end{pmatrix}}_{\triangleq \mb J} \mb z_k - \begin{pmatrix}
\Gamma_k \phi_{\tau}\mb v(t) \\
0
\end{pmatrix}  + \frac{\diff \mb R_k}{\diff t} - \mb J \mb R_k(t).
\label{eq_SI:z_diff}
\end{equation}
Furthermore, we note that $\mb R_{k}(t)$ can be modeled as a generalized Ornstein–Uhlenbeck process and $\frac{\diff \mb R_k}{\diff t} - \mb J \mb R_k(t)$ can be represented by
\begin{equation}
\frac{\diff \mb R_k}{\diff t} - \mb J \mb R_k(t) = \Lambda_k \dot{\mb W}_{k,t},
\label{eq_SI:OU_process}
\end{equation}
where $\dot{\mb W}_{k,t}$ is the standard white noise and $\Lambda_k\Lambda_k^T = -k_BT (\mb J + \mb J^T)$. With this choice, the covariance of $\mb R_k(t) =  \left[\mb R_{k,1}, \mb R_{k,2}\right]$ is given by
\begin{equation*}
\left\langle \mb R_k(t) \mb R_k(\tau)^T\right\rangle = k_BT {\rm e}^{-\alpha (t-\tau)}  \begin{pmatrix}
\cos(\omega_k (t-\tau))I &    \sin(\omega_k (t-\tau))I \\
-\sin(\omega_k (t-\tau))I & \cos(\omega_k (t-\tau))I
\end{pmatrix}     
\end{equation*}
such that Eq. \eqref{eq_SI:covariance_R_1} remains valid.
Using Eqs. \eqref{eq_SI:z_12}\eqref{eq_SI:p_dynamics}\eqref{eq_SI:z_diff}, we can write the reduced model \eqref{eq_SI:MZ_ansatz} in the form of
\begin{equation}
\begin{split}
\frac{\diff }{\diff t}\begin{pmatrix} \mb q \\ \mb p \\ \cdots \\ \mb z_{k,1} \\ \mb z_{k,2} \\ \cdots \end{pmatrix} 
&= \begin{pmatrix}
0 & I & \cdots &0 &0 &\cdots \\
-I & 0                & \cdots & \phi(\mb q)^T \Gamma_k^T & 0      & \cdots \\ 
0 & \cdots           & \cdots & \cdots                & \cdots & \cdots \\ 
0  &-\Gamma_k \phi(\mb q) & \cdots & -\alpha I                & -\omega_k I & \cdots \\ 
0 &0                & \cdots & \omega_k I                 & -\alpha I & \cdots \\
0& \cdots           & \cdots & \cdots                & \cdots & \cdots \\
\end{pmatrix} 
\begin{pmatrix} \nabla U(\mb q) \\ \mb v \\ \cdots \\ \mb z_{k,1} \\ \mb z_{k,2} \\ \cdots \end{pmatrix} +
 \begin{pmatrix} 0 \\ 0 \\ \cdots \\ \Lambda_{k} \dot{\mb W}_{k,t} \\  \cdots \end{pmatrix} \\
 &\triangleq \mb K \nabla F(\mb q, \mb p, \cdots, \mb z_{k,1}, \mb z_{k,2}, \cdots) + \Lambda \dot{\mb W}_t,\\
\end{split}
\label{eq_SI:reduced_model_Markovian}
\end{equation}
where $\mb K$ is the first matrix of the right-hand-side of Eq. \eqref{eq_SI:reduced_model_Markovian}, $F(\mb q, \mb p, \cdots, \mb z_{k,1}, \mb z_{k,2}, \cdots) = U(\mb q) + \frac{1}{2}\mb p^T \mb M^{-1}\mb p +  \frac{1}{2} \sum_{k=1}^{N_\omega} \left(\mb z_{k,1}^T\mb z_{k,1} + \mb z_{k,2}^T\mb z_{k,2}\right)$ is the total free energy of the extended system, and $\Lambda = {\rm diag}(0, 0, \cdots, \Lambda_k, \cdots)$. 
Using \eqref{eq_SI:OU_process}, it is easy to show $\Lambda \Lambda^T = -k_BT (\mb K+\mb K^T)$. Therefore, the gradient system \eqref{eq_SI:reduced_model_Markovian} (i.e., the reduced model \eqref{eq_SI:MZ_ansatz}) has the invariant density function
\begin{equation*}
\rho_{\rm eq} (\mb q, \mb p, \mb z) = \exp\left[-F(\mb q, \mb p, \mb z)/k_BT\right].    
\end{equation*}
\end{proof}

\section{Training details of the reduced model}
We represent the state encoders $\phi(\mb q) = \left[\phi_1(\mb q), \cdots, \phi_n(\mb q)\right]$ and $\Theta(t)$ in reduced model \eqref{eq_SI:MZ_ansatz} by 
\begin{equation}
\begin{split}
\phi_i(\mb q) &=  \mb H_i^T \psi (\mb q) \\
\Theta(t) &= {\rm e}^{-\alpha t} \sum_{k=0}^{N_\omega} \hat{\Theta}_k \cos(\omega_k t),
\end{split}
\label{eq:phi_theta}
\end{equation}
where $\psi(\mb q) = \left[\psi_1(\mb q), \cdots, \psi_{N_b}(\mb q)\right]$ is a set of sparse bases, $\mb H = \left[\mb H_1^T, \cdots, \mb H_n^T\right]$ are trainable coefficients. 
$\omega_k = 2\pi k/T_c$ is the frequency, 
where $T_c$ is the time domain cut-off of the kernel and $\delta t$ is the step size of the discrete samples. \textcolor{black}{For the present study, $\psi$ is chosen as the uniform piecewise linear basis function defined on $[2.8, 4.1]$ with $N_b = 66$, $T_c=200$ and $N_\omega = 2000$.}

To train the reduced model, we use the correlation functions conditional with differential initial states such that the state-dependent nature can be manifested. Specifically, we right-multiply the second equation of Eq. \eqref{eq_SI:MZ_ansatz} by $\mb v(0)$ and take the conditional expectation on $\mb q_0 = \mb q^{\ast}$, which defines $\mb g(t; \mb q^{\ast}) = \left\langle [\dot{\mb p}_t + \nabla U(\mb q_t)] \mb v_0^T\vert \mb q_0 = \mb q^{\ast}\right\rangle$. Accordingly, we define the empirical loss function
\begin{equation}
\begin{split}
L &= \sum_{l=1}^{N_q}\sum_{k=1}^{N_t} \left\Vert  \widetilde{\mb g} (t_k; \mb q^{(l)}) - \mb g(t_k; \mb q^{(l)})  \right\Vert^2, \\
\widetilde{\mb g}(t_k; \mb q^{(l)}) &= \sum_{j=1}^k {\rm Tr} \left[\Theta(t_k-t_j) \mb H \mb C_{\psi,\psi}(t_k, t_j; \mb q^{(l)}) \mb H^T \right] \delta t,
\end{split}
\nonumber
\end{equation}
where $\widetilde{\mb g}(\cdot)$ represents the prediction by the reduced model; $l$ and $k$ correspond to various initial states and discrete time, respectively.  
$\mb C_{\psi,\psi}(t, \tau; \mb q^{\ast}) = \left\langle \psi_\tau \mb v_\tau \mb v_0^T \psi_t^T \vert \mb q_0 = \mb q^{\ast}\right\rangle$
is a three-point correlation characterizing the coupling among the bases. 

Besides the conditional correlation functions, we can also introduce the loss function with respect to the overall correlation function, i.e., 
\begin{equation}
\begin{split}
L_2 &= \sum_{k=1}^{N_t} \left\Vert  \widetilde{\mb g}_2 (t_k) - \mb g_2(t_k)  \right\Vert^2 \\
\widetilde{\mb g}_2(t_k) &= \sum_{j=1}^k {\rm Tr} \left[\Theta(t_k-t_j) \mb H \overline{\mb C}_{\psi,\psi}(t_k, t_j) \mb H^T \right] \delta t,
\end{split}
\nonumber
\end{equation}
where  $\mb g_2(t) = \left\langle [\dot{\mb p}_t + \nabla U(\mb q_t)] \mb v_0^T \right\rangle$ is the overall correlation and $\overline{\mb C}_{\psi,\psi}(t, \tau) =  \left\langle \psi_\tau \mb v_\tau \mb v_0^T \psi_t^T \right\rangle$. In particular, if there is scale separation between $c_{vv}(t)$ and $c_{qq}(t)$ (e.g., $c_{vv}(t)$ decays much faster than $c_{qq}(t)$; see Fig. 2 and Fig. 4), we may approximate $\overline{\mb C}_{\psi,\psi}(t, \tau)$ by two-point correlations, i.e., $\overline{\mb C}_{\psi,\psi}(t, \tau) \approx \left\langle \psi_\tau \otimes \psi_t^T \right\rangle : \left\langle \mb v_\tau \mb v_0^T \right\rangle$. 

Efficient training is achieved by using the following numerical methods to evaluate $\widetilde{\mb g}$ and $\widetilde{\mb g}_2$. 
Specifically, $\psi_\tau \psi_t^T$ can be efficiently pre-computed with $O(1)$ complexity by using the sparse piecewise linear basis functions. Furthermore, we can use the low-rank representation (e.g., based on the singular value decomposition) of $\mb C_{\psi,\psi}$ and $\overline{\mb C}_{\psi,\psi}$ to accelerate the matrix production $\mb H \mb C_{\psi,\psi} \mb H^T$. In addition, the convolution on index $j$ can be efficiently evaluated by the Fast Fourier Transform algorithm \cite{Coole_Tukey_FFT_1965}.

While $L_2$ alone is insufficient to characterize the emergence of the state-dependent memory, it serves as a necessary condition and can facilitate the learning of the reduced model.  In practice, we can use both loss functions to train the reduced model with $N_q=65$, $N_t=300$ for $L$, and $N_t=30000$ for $L_2$. Specifically, the training is conducted by the Adam \cite{Kingma_Ba_Adam_2015} optimization method in three stages with 2000, 6000, and 6000 steps respectively. For the first stage, we only use $L_2$ to train the model with a constant learning rate of $0.04$. For each step of the following two stages, 16 initial states (i.e., $q^{(l)}$) are randomly selected as one training batch to evaluate the total loss $L_t = L + L_2$. For both stages, the initial learning rate is $1\times 10^{-2}$ and the exponential decay rate is $0.9$ per 150 steps.

\section{Simulation of the reduced model}
To simulate the reduced model \eqref{eq_SI:MZ_ansatz}, we follow Prop. \ref{prop:markovian_form} and generate $\widetilde{\mb R}(t)$ on $[0, T]$ similar to Refs. \cite{berkowitz1983generalized, ogorodnikov1996numerical} by   
\begin{equation}
\widetilde{\mb R}(t) = \beta^{-1/2} \sum_{k=0}^{2N} \widetilde{\Theta}_k^{1/2} \left[\cos(\omega_k t) \xi_k + \sin(\omega_k t) \eta_k\right],
\label{eq_SI:random_noise}
\end{equation}
where $\beta^{-1} = k_BT$, $\widetilde{\Theta}_k$ are the Fourier (essentially cosine) modes of $\Theta(|t|)$ on $\left[-T, T\right]$, $\xi_k$ and $\eta_k$ are independent Gaussian random vectors, and $N$ is the total number of simulation step.  

Specifically, for large simulation time $T$, the Fourier modes $\widetilde{\Theta}_k$ is given by 
\begin{equation*}
\begin{split}
\widetilde{\Theta}_k 
&=  \sum_{j=1}^{N_\omega} \int_{-T}^T  {\rm e}^{-\alpha \vert t\vert} \hat{\Theta}_j \cos(\omega_j t) \cos(\omega_k t) \diff t \\
&= \sum_{j=1}^{N_\omega} \int_{0}^T  {\rm e}^{-\alpha  t} \hat{\Theta}_j \left( \cos\left((\omega_j - \omega_k)t\right) +  \cos\left(\omega_j + \omega_k)t\right) \right) \\
&\approx \sum_{j=1}^{N_\omega}\left(\frac{\alpha \hat{\Theta}_j}{\alpha^2 + (\omega_j - \omega_k)^2} + \frac{\alpha \hat{\Theta}_j}{\alpha^2 + (\omega_j + \omega_k)^2}\right).
\end{split}
\end{equation*}
Therefore $\widetilde{\mb R}(t)$ can be generated using the Fast Fourier Transform algorithm \cite{Coole_Tukey_FFT_1965} using $O(N\log N)$ complexity. Also, the convolution term $\int_0^t \phi_t^T \Theta(t-\tau) \phi_{\tau} \mb v(\tau) \diff \tau$ in Eq. \eqref{eq_SI:MZ_ansatz} can be computed using the fast convolution method developed in Ref. \cite{Achim_SIAM_2006} with $O(N\log N)$ complexity.

\section{Full atomistic model}
\label{sec_SI:MD}
We consider the full micro-scale model of benzyl bromide (see Fig. \ref{fig:mol_blb} for a sketch of the molecule structure) in an aqueous environment. The general AMBER \cite{Wang_Wolf_Amber_JCC_2004} force field is used for the benzyl bromide molecule and the partial charges of molecule atoms were set by the restrained electrostatic potential (RESP) approach \cite{Bayly_Cieplak_RESP_JPC_1993}. The rigid TIP3P water model \cite{Jorgensen_Chandrasekhar_TIP3P_JCP_1983} is used for the water molecules and the bond lengths and angles were held constant through the SHAKE algorithm \cite{Ryckaert_Ciccotti_SHAKE_JCP_1977, Miyamoto_Kollman_Settle_JCC_2004}. Long-range electrostatic interactions were calculated using a Particle Mesh Ewald summation with a relative error set to be $10^{-4}$. 
The full system consists of one benzyl bromide molecule and 2400 water molecules with the periodic boundary condition imposed along each direction. The isothermal-isobaric thermostat \cite{Martyna_Klein_JCP_1994} is used to equilibrate the system for 16 ns at 298K and 1 bar using a time step of $1$ fs. Following the equilibration, the box size is scaled to be near $41.5 \times 41.5 \times 41.5$ {\AA}$^3$. The simulation was run for a production period of 2.5 $\mu$s in a canonical ensemble with a Nos\'e-Hoover thermostat \cite{Nose_Mol_Phys_1984, Hoover1985}. 

\begin{figure}
    \centering
    \includegraphics[width=0.36\textwidth]{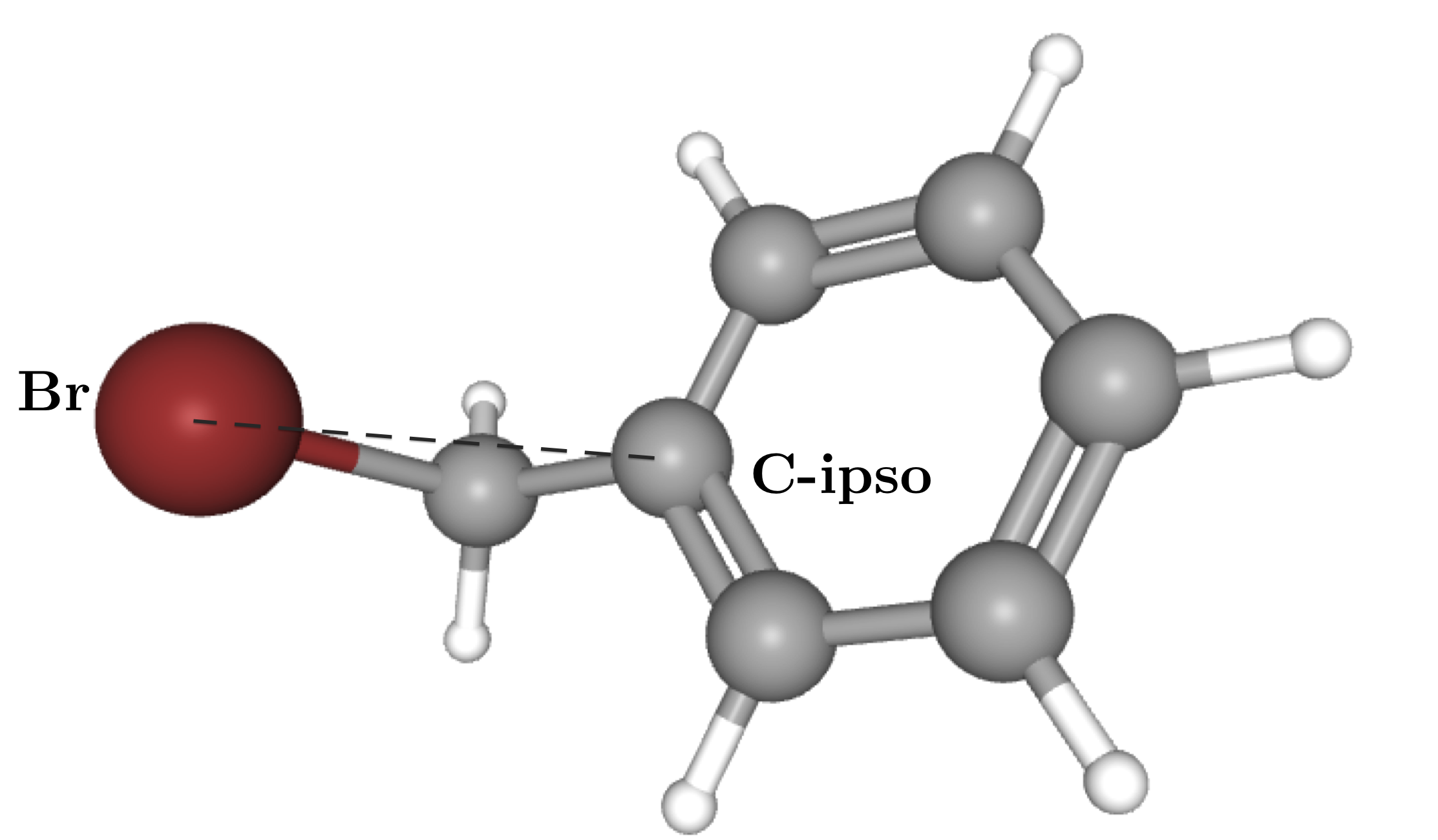}
    \caption{A sketch of the molecule benzyl bromide. The resolved variable is defined as the distance between the bromine atom and the ipso-carbon atom.}
    \label{fig:mol_blb}
\end{figure}

The resolved variable $q$ is defined as the distance between
the bromine atom and the ipso-carbon atom. The free energy is obtained from the probability density function $\rho(q)$ (see the inset plot of Fig. 2(a)), i.e., $U(q) = -k_BT\ln \rho(q)$, where $\rho(q)$ is directly obtained from the full MD samples using the kernel density estimation. To verify the accuracy of the constructed $U(q)$, we calculate the expectation of $q\nabla U(q)$ on the sample. The numerical result gives $0.996 k_BT$ and is close to the theoretical prediction $\left \langle q \nabla U(q)\right\rangle = \int q \nabla U(q) {\rm e}^{-U(q)/k_BT} \diff q \equiv k_BT$.

\section{Additional numerical results}
\subsection{Limitations of the standard GLE near the local minima}
Fig. \ref{fig_SI:cvv_cvq_local_minima} shows the predictions of the conditional correlations $c_{qv}(t, q^{\ast})$ and $c_{vv}(t, q^{\ast})$  for $q^{\ast}$ representing the two local minima. Similar to the results of the saddle point as shown in Fig. 3(b), the predictions of the standard GLE show apparent deviations from the full MD results due to the ignorance of the state-dependent memory nature. In contrast, the predictions of the present model with four features can accurately recover the MD predictions.

\begin{figure}[htbp]
    \centering
    \includegraphics[scale=0.25]{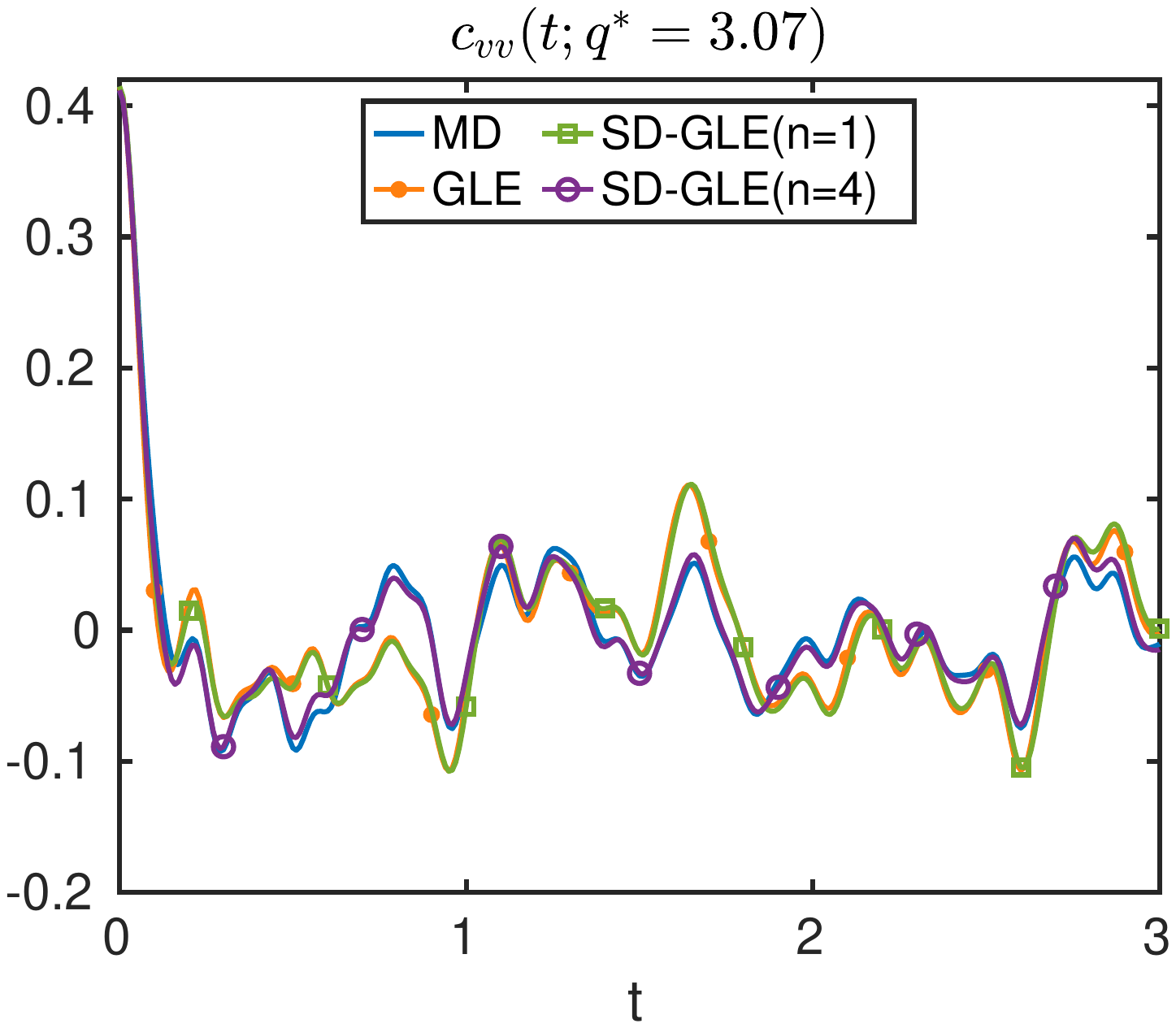}
    \includegraphics[scale=0.25]{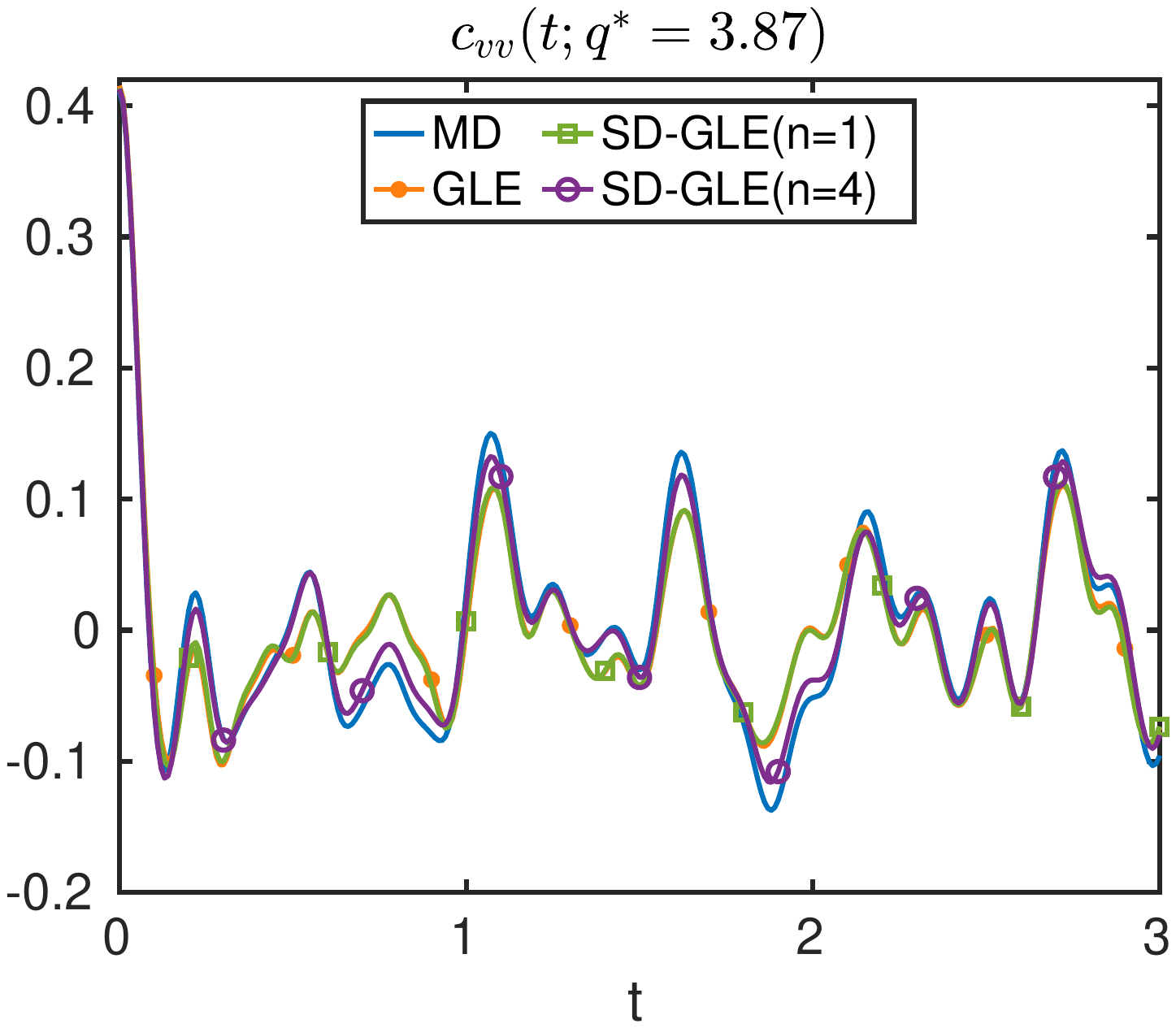}
    \includegraphics[scale=0.25]{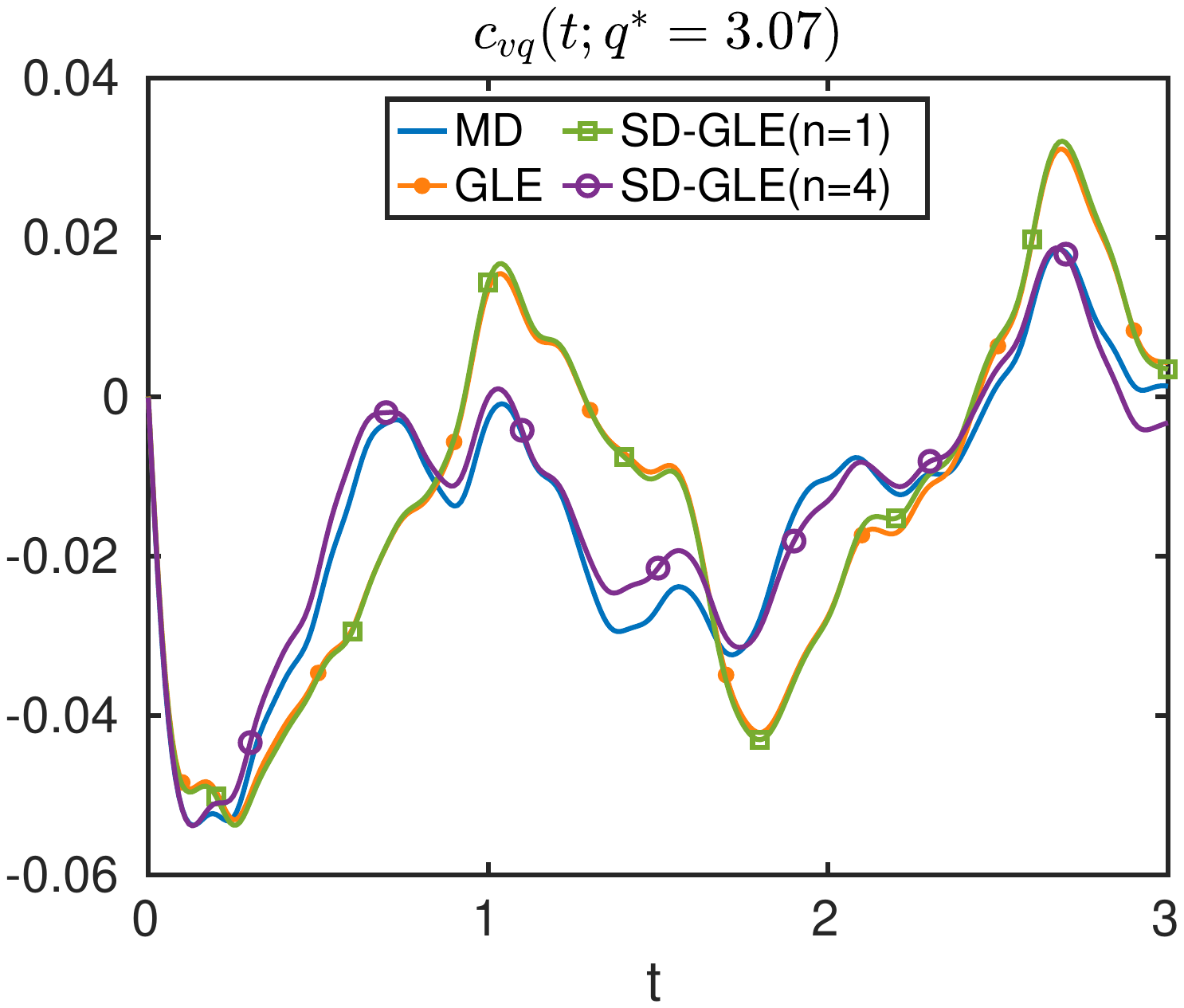}
    \includegraphics[scale=0.25]{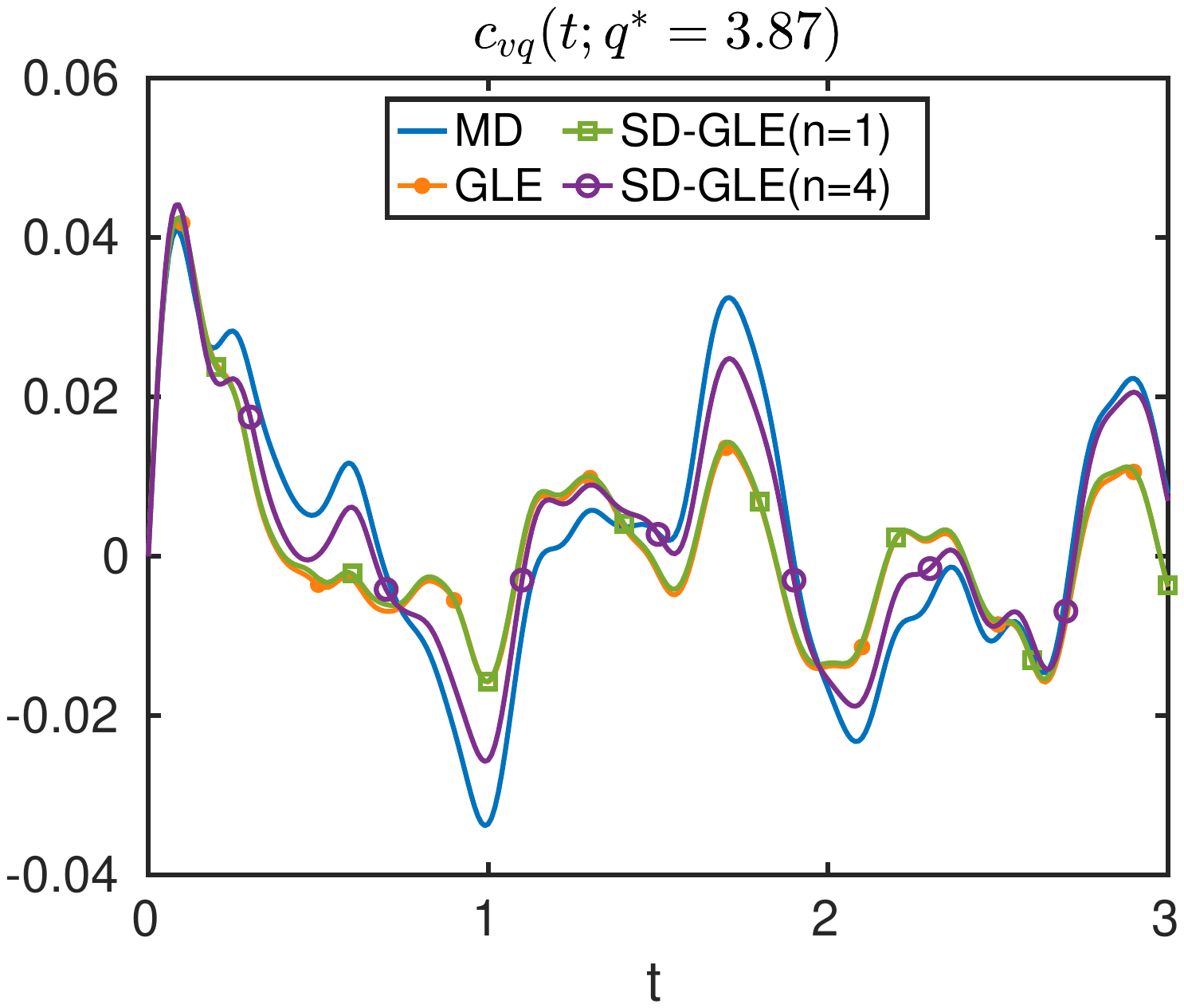}

    \caption{The conditional correlation functions $c_{qv}(t, q^{\ast})$ and $c_{vv}(t, q^{\ast})$ for the two local minima predicted by the full MD, the standard GLE, and the present model (SD-GLE) constructed using one and four spatio-features. Left: $q^{\ast} = 3.07$; Right:  $q^{\ast} = 3.87$. The predictions by the standard GLE show apparent discrepancies with the full MD results. }
    \label{fig_SI:cvv_cvq_local_minima}
\end{figure}

\subsection{Other forms of the reduced model}
For comparison, 
we also consider other forms of the reduced model. In particular, we retain the encoders $\phi$ in Eq. \eqref{eq_SI:MZ_ansatz} but set $\Theta(t)$ to be diagonal, i.e., we ignore the non-Markovian coupling among the different state features. The reduced model is trained using four features. 
Fig. \ref{fig_SI:SD-diag} shows the conditional correlation $c_{vv}(t; q^{\ast})$ obtained from the full MD and different reduced models. 
The prediction of the constructed model (labeled by ``SD-GLE-Diag'') shows apparent deviations from the full MD result with incremental improvement over the stand GLE. The large discrepancy reveals the complex state-dependent nature; the non-Markovian effect can be neither approximated by ansatz like $\gamma(q)\theta(t)$ as a simple generalization/re-scaling of the initial value at $t=0$, nor represented by the coupling with the independent bath variables. Instead, the non-Markovian coupling among the various state-features retained in the present model plays a crucial for accurately modeling the heterogeneous energy dissipation arising from the unresolved intramolecular interactions and reproducing the collective dynamics.  

\begin{figure}
    \centering
    \includegraphics[scale=0.3]{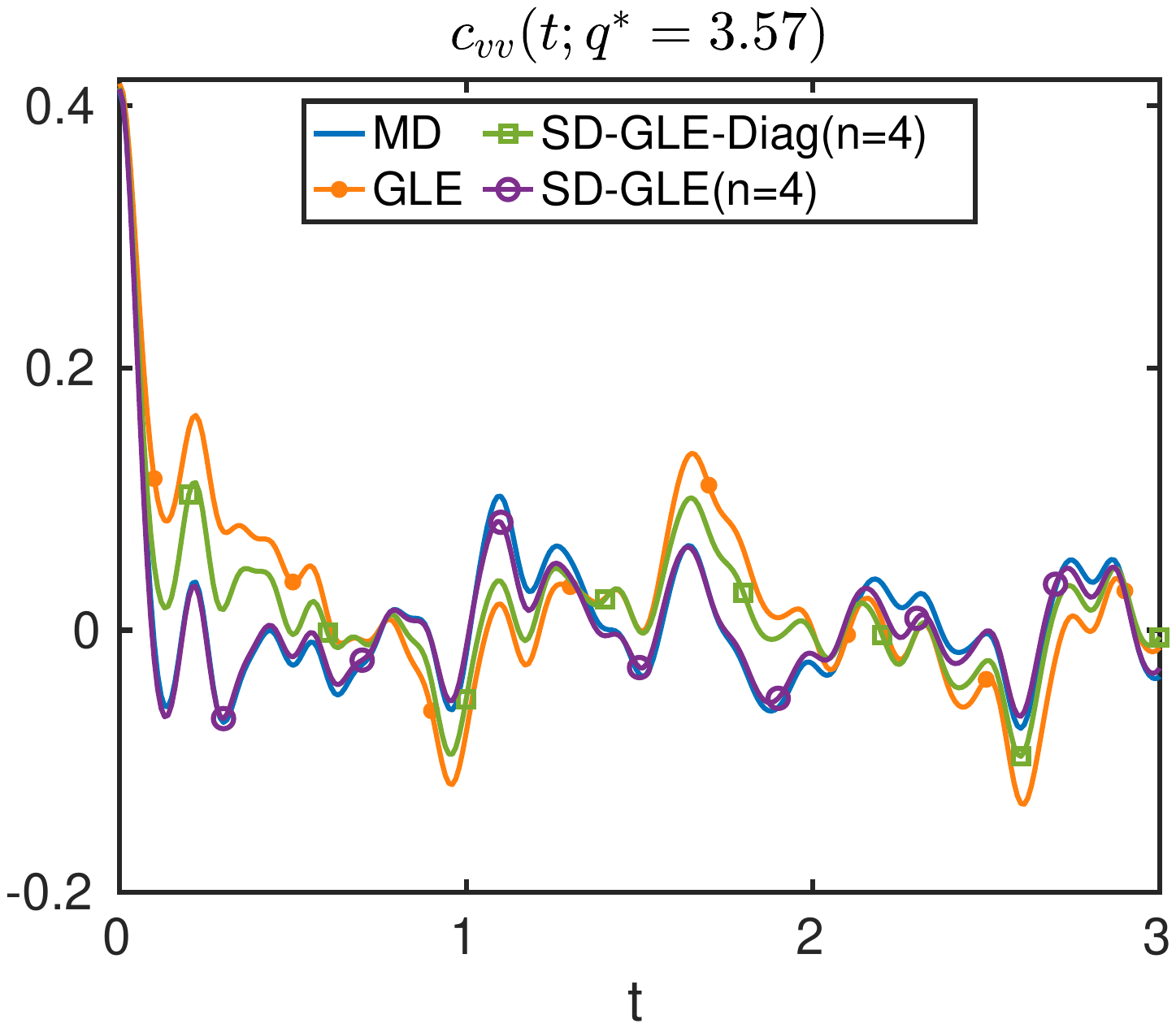}
    \caption{The conditional correlation functions $c_{vv}(t, q^{\ast})$ for the saddle point predicted by the full MD, the standard GLE, the reduced models constructed using four state-features with diagonal $\Theta(t)$ (SD-GLE-Diag) and full $\Theta(t)$ (SD-GLE). The large discrepancy between the SD-GLE-Diag model and the full MD results implies the complexity of the state-dependency of the memory term, which can not be well represented by the coupling of independent bath variables. The non-Markovian interactions among the state-features are essential to capture the heterogeneous energy dissipation process.}
    \label{fig_SI:SD-diag}
\end{figure}

\section{Generalization of the present reduced model formulation}
So far, we have constructed the reduced model \eqref{eq_SI:MZ_ansatz} by assuming the matrix-valued kernel $\Theta(t)$ is symmetry.  In fact, this form can be generalized by introducing an anti-symmetry part, i.e., 
\begin{equation}
\Theta(t) = {\rm e}^{-\alpha t} \sum_{k=0}^{N_\omega} (\Gamma_{k,1}^T\Gamma_{k,1}+ \Gamma_{k,2}^T\Gamma_{k,2}) \cos(\omega_k t)+(\Gamma_{k,1}^T\Gamma_{k,2}- \Gamma_{k,2}^T\Gamma_{k,1}) \sin(\omega_k t),    
\label{eq_SI:Theta_general}
\end{equation}
where $\Gamma_{k,1}$ and $\Gamma_{k,2}$ are lower-triangular matrices representing the Fourier modes of $\Theta(t)$. The form is general non-symmetric except for $t=0$ and satisfies $\Theta(-t) = \Theta(t)^T$.

Similar to the symmetry form, we can model the fluctuation term $\bm{\mathcal{R}}_t$ as a noise in the form of 
$\bm{\mathcal{R}}_t = \phi_t^T \widetilde{\mb R}(t)$, where $\widetilde{\mb R}(t)$ is a Gaussian random process satisfying $\langle \widetilde{\mb R}(t) \widetilde{\mb R}(\tau)^T\rangle = k_BT{\rm e}^{-\alpha (t-\tau)}\Theta(t-\tau)$. Similar to Prop. \ref{prop:markovian_form}, we can show that this choice retains a consistent invariant density function. 
In practice, we can generate the noise term $\widetilde{\mb R}(t)$ on $[0, T]$ by
\begin{equation}
\begin{split}
\widetilde{\mb R}(t) &= \beta^{-1/2} \sum_{k=0}^{2N} \left[ \widetilde{\Theta}_{k,1}^{1/2} \cos(\omega_k t) \xi_k + \sin(\omega_k t) (Q_{1,k}^{1/2} \xi_k+ Q_{2,k}^{1/2}\eta_k )\right], \nonumber\\
Q_{1,k} &= \widetilde{\Theta}_{k,2} \widetilde{\Theta}_{k,1}^{-1} \widetilde{\Theta}_{k,2}^T \quad 
Q_{2,k} = \widetilde{\Theta}_{k,1} - \widetilde{\Theta}_{k,2} \widetilde{\Theta}_{k,1}^{-1} \widetilde{\Theta}_{k,2}^T,
\label{eq_SI:random_noise_general}
\end{split}
\end{equation}
where $\beta^{-1} = k_BT$, $\widetilde{\Theta}_{k,1}$, $\widetilde{\Theta}_{k,2}$ are the Fourier cosine and sine modes on $\left[-T, T\right]$ with $\Theta(-t)=\Theta(t)^T$, $\xi_k$ and $\eta_k$ are independent Gaussian random vectors, and $N$ is the total number of simulation step. 
Here $\widetilde{\mb R}(t)$ can be also generated using the Fast Fourier Transform algorithm \cite{Coole_Tukey_FFT_1965} using $O(N\log N)$ complexity. 
We will investigate this generalized formulation for model reduction in future studies.


\begin{thebibliography}{82}%
\makeatletter
\providecommand \@ifxundefined [1]{%
 \@ifx{#1\undefined}
}%
\providecommand \@ifnum [1]{%
 \ifnum #1\expandafter \@firstoftwo
 \else \expandafter \@secondoftwo
 \fi
}%
\providecommand \@ifx [1]{%
 \ifx #1\expandafter \@firstoftwo
 \else \expandafter \@secondoftwo
 \fi
}%
\providecommand \natexlab [1]{#1}%
\providecommand \enquote  [1]{``#1''}%
\providecommand \bibnamefont  [1]{#1}%
\providecommand \bibfnamefont [1]{#1}%
\providecommand \citenamefont [1]{#1}%
\providecommand \href@noop [0]{\@secondoftwo}%
\providecommand \href [0]{\begingroup \@sanitize@url \@href}%
\providecommand \@href[1]{\@@startlink{#1}\@@href}%
\providecommand \@@href[1]{\endgroup#1\@@endlink}%
\providecommand \@sanitize@url [0]{\catcode `\\12\catcode `\$12\catcode
  `\&12\catcode `\#12\catcode `\^12\catcode `\_12\catcode `\%12\relax}%
\providecommand \@@startlink[1]{}%
\providecommand \@@endlink[0]{}%
\providecommand \url  [0]{\begingroup\@sanitize@url \@url }%
\providecommand \@url [1]{\endgroup\@href {#1}{\urlprefix }}%
\providecommand \urlprefix  [0]{URL }%
\providecommand \Eprint [0]{\href }%
\providecommand \doibase [0]{http://dx.doi.org/}%
\providecommand \selectlanguage [0]{\@gobble}%
\providecommand \bibinfo  [0]{\@secondoftwo}%
\providecommand \bibfield  [0]{\@secondoftwo}%
\providecommand \translation [1]{[#1]}%
\providecommand \BibitemOpen [0]{}%
\providecommand \bibitemStop [0]{}%
\providecommand \bibitemNoStop [0]{.\EOS\space}%
\providecommand \EOS [0]{\spacefactor3000\relax}%
\providecommand \BibitemShut  [1]{\csname bibitem#1\endcsname}%
\let\auto@bib@innerbib\@empty
\bibitem [{\citenamefont {Koopman}(1931)}]{Koopman315}%
  \BibitemOpen
  \bibfield  {author} {\bibinfo {author} {\bibfnamefont {B.~O.}\ \bibnamefont
  {Koopman}},\ }\href@noop {} {\bibfield  {journal} {\bibinfo  {journal}
  {Proceedings of the National Academy of Sciences}\ }\textbf {\bibinfo
  {volume} {17}},\ \bibinfo {pages} {315} (\bibinfo {year} {1931})}\BibitemShut
  {NoStop}%
\bibitem [{\citenamefont {Mori}(1965)}]{Mori1965}%
  \BibitemOpen
  \bibfield  {author} {\bibinfo {author} {\bibfnamefont {H.}~\bibnamefont
  {Mori}},\ }\href@noop {} {\bibfield  {journal} {\bibinfo  {journal} {Progress
  of Theoretical Physics}\ }\textbf {\bibinfo {volume} {33}},\ \bibinfo {pages}
  {423} (\bibinfo {year} {1965})}\BibitemShut {NoStop}%
\bibitem [{\citenamefont {Zwanzig}(1961)}]{Zwanzig61}%
  \BibitemOpen
  \bibfield  {author} {\bibinfo {author} {\bibfnamefont {R.}~\bibnamefont
  {Zwanzig}},\ }\href@noop {} {\bibfield  {journal} {\bibinfo  {journal}
  {Lectures in Theoretical Physics}\ }\textbf {\bibinfo {volume} {3}},\
  \bibinfo {pages} {106} (\bibinfo {year} {1961})}\BibitemShut {NoStop}%
\bibitem [{\citenamefont {Zwanzig}(2001)}]{Zwanzigbook}%
  \BibitemOpen
  \bibfield  {author} {\bibinfo {author} {\bibfnamefont {R.}~\bibnamefont
  {Zwanzig}},\ }\href@noop {} {\emph {\bibinfo {title} {Nonequilibrium
  Statistical Mechanics}}}\ (\bibinfo  {publisher} {Oxford University Press},\
  \bibinfo {year} {2001})\BibitemShut {NoStop}%
\bibitem [{\citenamefont {Lange}\ and\ \citenamefont
  {Grubm{\"u}ller}(2006)}]{lange2006collective}%
  \BibitemOpen
  \bibfield  {author} {\bibinfo {author} {\bibfnamefont {O.~F.}\ \bibnamefont
  {Lange}}\ and\ \bibinfo {author} {\bibfnamefont {H.}~\bibnamefont
  {Grubm{\"u}ller}},\ }\href@noop {} {\bibfield  {journal} {\bibinfo  {journal}
  {J. Chem. Phys.}\ }\textbf {\bibinfo {volume} {124}},\ \bibinfo {pages}
  {214903} (\bibinfo {year} {2006})}\BibitemShut {NoStop}%
\bibitem [{\citenamefont {Darve}\ \emph {et~al.}(2009)\citenamefont {Darve},
  \citenamefont {Solomon},\ and\ \citenamefont {Kia}}]{Darve_PNAS_2009}%
  \BibitemOpen
  \bibfield  {author} {\bibinfo {author} {\bibfnamefont {E.}~\bibnamefont
  {Darve}}, \bibinfo {author} {\bibfnamefont {J.}~\bibnamefont {Solomon}}, \
  and\ \bibinfo {author} {\bibfnamefont {A.}~\bibnamefont {Kia}},\ }\href@noop
  {} {\bibfield  {journal} {\bibinfo  {journal} {Proc. Natl. Acad. Sci.}\
  }\textbf {\bibinfo {volume} {106}},\ \bibinfo {pages} {10884} (\bibinfo
  {year} {2009})}\BibitemShut {NoStop}%
\bibitem [{\citenamefont {Ceriotti}\ \emph {et~al.}(2009)\citenamefont
  {Ceriotti}, \citenamefont {Bussi},\ and\ \citenamefont
  {Parrinello}}]{ceriotti2009langevin}%
  \BibitemOpen
  \bibfield  {author} {\bibinfo {author} {\bibfnamefont {M.}~\bibnamefont
  {Ceriotti}}, \bibinfo {author} {\bibfnamefont {G.}~\bibnamefont {Bussi}}, \
  and\ \bibinfo {author} {\bibfnamefont {M.}~\bibnamefont {Parrinello}},\
  }\href@noop {} {\bibfield  {journal} {\bibinfo  {journal} {Physical review
  letters}\ }\textbf {\bibinfo {volume} {102}},\ \bibinfo {pages} {020601}
  (\bibinfo {year} {2009})}\BibitemShut {NoStop}%
\bibitem [{\citenamefont {Baczewski}\ and\ \citenamefont
  {Bond}(2013)}]{baczewski2013numerical}%
  \BibitemOpen
  \bibfield  {author} {\bibinfo {author} {\bibfnamefont {A.~D.}\ \bibnamefont
  {Baczewski}}\ and\ \bibinfo {author} {\bibfnamefont {S.~D.}\ \bibnamefont
  {Bond}},\ }\href@noop {} {\bibfield  {journal} {\bibinfo  {journal} {The
  Journal of chemical physics}\ }\textbf {\bibinfo {volume} {139}},\ \bibinfo
  {pages} {044107} (\bibinfo {year} {2013})}\BibitemShut {NoStop}%
\bibitem [{\citenamefont {Davtyan}\ \emph {et~al.}(2015)\citenamefont
  {Davtyan}, \citenamefont {Dama}, \citenamefont {Voth},\ and\ \citenamefont
  {Andersen}}]{Dav_Voth_JCP_2015}%
  \BibitemOpen
  \bibfield  {author} {\bibinfo {author} {\bibfnamefont {A.}~\bibnamefont
  {Davtyan}}, \bibinfo {author} {\bibfnamefont {J.~F.}\ \bibnamefont {Dama}},
  \bibinfo {author} {\bibfnamefont {G.~A.}\ \bibnamefont {Voth}}, \ and\
  \bibinfo {author} {\bibfnamefont {H.~C.}\ \bibnamefont {Andersen}},\
  }\href@noop {} {\bibfield  {journal} {\bibinfo  {journal} {J. Chem. Phys.}\
  }\textbf {\bibinfo {volume} {142}},\ \bibinfo {eid} {154104} (\bibinfo {year}
  {2015})}\BibitemShut {NoStop}%
\bibitem [{\citenamefont {Lei}\ \emph {et~al.}(2016)\citenamefont {Lei},
  \citenamefont {Baker},\ and\ \citenamefont {Li}}]{Lei_Li_PNAS_2016}%
  \BibitemOpen
  \bibfield  {author} {\bibinfo {author} {\bibfnamefont {H.}~\bibnamefont
  {Lei}}, \bibinfo {author} {\bibfnamefont {N.~A.}\ \bibnamefont {Baker}}, \
  and\ \bibinfo {author} {\bibfnamefont {X.}~\bibnamefont {Li}},\ }\href@noop
  {} {\bibfield  {journal} {\bibinfo  {journal} {Proc. Natl. Acad. Sci.}\
  }\textbf {\bibinfo {volume} {113}},\ \bibinfo {pages} {14183} (\bibinfo
  {year} {2016})}\BibitemShut {NoStop}%
\bibitem [{\citenamefont {Russo}\ \emph {et~al.}(2019)\citenamefont {Russo},
  \citenamefont {Dur{\'a}n-Olivencia}, \citenamefont {Kevrekidis},\ and\
  \citenamefont {Kalliadasis}}]{russo2019deep}%
  \BibitemOpen
  \bibfield  {author} {\bibinfo {author} {\bibfnamefont {A.}~\bibnamefont
  {Russo}}, \bibinfo {author} {\bibfnamefont {M.~A.}\ \bibnamefont
  {Dur{\'a}n-Olivencia}}, \bibinfo {author} {\bibfnamefont {I.~G.}\
  \bibnamefont {Kevrekidis}}, \ and\ \bibinfo {author} {\bibfnamefont
  {S.}~\bibnamefont {Kalliadasis}},\ }\href@noop {} {\bibfield  {journal}
  {\bibinfo  {journal} {arXiv preprint arXiv:1903.09562}\ } (\bibinfo {year}
  {2019})}\BibitemShut {NoStop}%
\bibitem [{\citenamefont {Jung}\ \emph {et~al.}(2017)\citenamefont {Jung},
  \citenamefont {Hanke},\ and\ \citenamefont {Schmid}}]{Jung_Hanke_JCTC_2017}%
  \BibitemOpen
  \bibfield  {author} {\bibinfo {author} {\bibfnamefont {G.}~\bibnamefont
  {Jung}}, \bibinfo {author} {\bibfnamefont {M.}~\bibnamefont {Hanke}}, \ and\
  \bibinfo {author} {\bibfnamefont {F.}~\bibnamefont {Schmid}},\ }\href@noop {}
  {\bibfield  {journal} {\bibinfo  {journal} {Journal of Chemical Theory and
  Computation}\ }\textbf {\bibinfo {volume} {13}},\ \bibinfo {pages} {2481}
  (\bibinfo {year} {2017})}\BibitemShut {NoStop}%
\bibitem [{\citenamefont {Lee}\ \emph {et~al.}(2019)\citenamefont {Lee},
  \citenamefont {Ahn},\ and\ \citenamefont {Darve}}]{Lee2019}%
  \BibitemOpen
  \bibfield  {author} {\bibinfo {author} {\bibfnamefont {H.~S.}\ \bibnamefont
  {Lee}}, \bibinfo {author} {\bibfnamefont {S.-H.}\ \bibnamefont {Ahn}}, \ and\
  \bibinfo {author} {\bibfnamefont {E.~F.}\ \bibnamefont {Darve}},\ }\href@noop
  {} {\bibfield  {journal} {\bibinfo  {journal} {The Journal of Chemical
  Physics}\ }\textbf {\bibinfo {volume} {150}},\ \bibinfo {pages} {174113}
  (\bibinfo {year} {2019})}\BibitemShut {NoStop}%
\bibitem [{\citenamefont {Ma}\ \emph {et~al.}(2019)\citenamefont {Ma},
  \citenamefont {Li},\ and\ \citenamefont {Liu}}]{ma2019coarse}%
  \BibitemOpen
  \bibfield  {author} {\bibinfo {author} {\bibfnamefont {L.}~\bibnamefont
  {Ma}}, \bibinfo {author} {\bibfnamefont {X.}~\bibnamefont {Li}}, \ and\
  \bibinfo {author} {\bibfnamefont {C.}~\bibnamefont {Liu}},\ }\href@noop {}
  {\bibfield  {journal} {\bibinfo  {journal} {Journal of Computational
  Physics}\ }\textbf {\bibinfo {volume} {380}},\ \bibinfo {pages} {170}
  (\bibinfo {year} {2019})}\BibitemShut {NoStop}%
\bibitem [{\citenamefont {Wang}\ \emph {et~al.}(2020)\citenamefont {Wang},
  \citenamefont {Ma},\ and\ \citenamefont {Pan}}]{wang2020data}%
  \BibitemOpen
  \bibfield  {author} {\bibinfo {author} {\bibfnamefont {S.}~\bibnamefont
  {Wang}}, \bibinfo {author} {\bibfnamefont {Z.}~\bibnamefont {Ma}}, \ and\
  \bibinfo {author} {\bibfnamefont {W.}~\bibnamefont {Pan}},\ }\href@noop {}
  {\bibfield  {journal} {\bibinfo  {journal} {Soft Matter}\ }\textbf {\bibinfo
  {volume} {16}},\ \bibinfo {pages} {8330} (\bibinfo {year}
  {2020})}\BibitemShut {NoStop}%
\bibitem [{\citenamefont {Zhu}\ and\ \citenamefont {Venturi}(2020)}]{Zhu2020}%
  \BibitemOpen
  \bibfield  {author} {\bibinfo {author} {\bibfnamefont {Y.}~\bibnamefont
  {Zhu}}\ and\ \bibinfo {author} {\bibfnamefont {D.}~\bibnamefont {Venturi}},\
  }\href@noop {} {\bibfield  {journal} {\bibinfo  {journal} {Journal of
  Statistical Physics}\ ,\ \bibinfo {pages} {1217}} (\bibinfo {year}
  {2020})}\BibitemShut {NoStop}%
\bibitem [{\citenamefont {Klippenstein}\ and\ \citenamefont {van~der
  Vegt}(2021)}]{Klippenstein_Vegt_JCP_2021}%
  \BibitemOpen
  \bibfield  {author} {\bibinfo {author} {\bibfnamefont {V.}~\bibnamefont
  {Klippenstein}}\ and\ \bibinfo {author} {\bibfnamefont {N.~F.~A.}\
  \bibnamefont {van~der Vegt}},\ }\href@noop {} {\bibfield  {journal} {\bibinfo
   {journal} {The Journal of Chemical Physics}\ }\textbf {\bibinfo {volume}
  {154}},\ \bibinfo {pages} {191102} (\bibinfo {year} {2021})}\BibitemShut
  {NoStop}%
\bibitem [{\citenamefont {Vroylandt}\ \emph {et~al.}(2022)\citenamefont
  {Vroylandt}, \citenamefont {Gouden{\`e}ge}, \citenamefont {Monmarch{\'e}},
  \citenamefont {Pietrucci},\ and\ \citenamefont
  {Rotenberg}}]{vroylandt2022likelihood}%
  \BibitemOpen
  \bibfield  {author} {\bibinfo {author} {\bibfnamefont {H.}~\bibnamefont
  {Vroylandt}}, \bibinfo {author} {\bibfnamefont {L.}~\bibnamefont
  {Gouden{\`e}ge}}, \bibinfo {author} {\bibfnamefont {P.}~\bibnamefont
  {Monmarch{\'e}}}, \bibinfo {author} {\bibfnamefont {F.}~\bibnamefont
  {Pietrucci}}, \ and\ \bibinfo {author} {\bibfnamefont {B.}~\bibnamefont
  {Rotenberg}},\ }\href@noop {} {\bibfield  {journal} {\bibinfo  {journal}
  {Proceedings of the National Academy of Sciences}\ }\textbf {\bibinfo
  {volume} {119}},\ \bibinfo {pages} {e2117586119} (\bibinfo {year}
  {2022})}\BibitemShut {NoStop}%
\bibitem [{\citenamefont {She}\ \emph {et~al.}(2023)\citenamefont {She},
  \citenamefont {Ge},\ and\ \citenamefont {Lei}}]{SheZ_JCP_2023}%
  \BibitemOpen
  \bibfield  {author} {\bibinfo {author} {\bibfnamefont {Z.}~\bibnamefont
  {She}}, \bibinfo {author} {\bibfnamefont {P.}~\bibnamefont {Ge}}, \ and\
  \bibinfo {author} {\bibfnamefont {H.}~\bibnamefont {Lei}},\ }\href@noop {}
  {\bibfield  {journal} {\bibinfo  {journal} {The Journal of Chemical Physics}\
  }\textbf {\bibinfo {volume} {158}},\ \bibinfo {pages} {034102} (\bibinfo
  {year} {2023})}\BibitemShut {NoStop}%
\bibitem [{\citenamefont {Xie}\ \emph {et~al.}(2022)\citenamefont {Xie},
  \citenamefont {Car},\ and\ \citenamefont {E}}]{xie2023gle}%
  \BibitemOpen
  \bibfield  {author} {\bibinfo {author} {\bibfnamefont {P.}~\bibnamefont
  {Xie}}, \bibinfo {author} {\bibfnamefont {R.}~\bibnamefont {Car}}, \ and\
  \bibinfo {author} {\bibfnamefont {W.}~\bibnamefont {E}},\ }\href@noop {}
  {\bibfield  {journal} {\bibinfo  {journal} {arXiv preprint arXiv:2211.06558}\
  } (\bibinfo {year} {2022})}\BibitemShut {NoStop}%
\bibitem [{\citenamefont {H{\"a}nggi}(1997)}]{Hanggi_Stochastic_book_1997}%
  \BibitemOpen
  \bibfield  {author} {\bibinfo {author} {\bibfnamefont {P.}~\bibnamefont
  {H{\"a}nggi}},\ }in\ \href@noop {} {\emph {\bibinfo {booktitle} {Stochastic
  Dynamics}}},\ \bibinfo {editor} {edited by\ \bibinfo {editor} {\bibfnamefont
  {L.}~\bibnamefont {Schimansky-Geier}}\ and\ \bibinfo {editor} {\bibfnamefont
  {T.}~\bibnamefont {P{\"o}schel}}}\ (\bibinfo  {publisher} {Springer Berlin
  Heidelberg},\ \bibinfo {address} {Berlin, Heidelberg},\ \bibinfo {year}
  {1997})\ pp.\ \bibinfo {pages} {15--22}\BibitemShut {NoStop}%
\bibitem [{\citenamefont {Klippenstein}\ \emph {et~al.}(2021)\citenamefont
  {Klippenstein}, \citenamefont {Tripathy}, \citenamefont {Jung}, \citenamefont
  {Schmid},\ and\ \citenamefont {van~der Vegt}}]{klippenstein2021introducing}%
  \BibitemOpen
  \bibfield  {author} {\bibinfo {author} {\bibfnamefont {V.}~\bibnamefont
  {Klippenstein}}, \bibinfo {author} {\bibfnamefont {M.}~\bibnamefont
  {Tripathy}}, \bibinfo {author} {\bibfnamefont {G.}~\bibnamefont {Jung}},
  \bibinfo {author} {\bibfnamefont {F.}~\bibnamefont {Schmid}}, \ and\ \bibinfo
  {author} {\bibfnamefont {N.~F.}\ \bibnamefont {van~der Vegt}},\ }\href@noop
  {} {\bibfield  {journal} {\bibinfo  {journal} {The Journal of Physical
  Chemistry B}\ }\textbf {\bibinfo {volume} {125}},\ \bibinfo {pages} {4931}
  (\bibinfo {year} {2021})}\BibitemShut {NoStop}%
\bibitem [{\citenamefont {{Posch}}\ \emph {et~al.}(1984)\citenamefont
  {{Posch}}, \citenamefont {{Balucani}},\ and\ \citenamefont
  {{Vallauri}}}]{Posch_Balucani_Physica_A_1984}%
  \BibitemOpen
  \bibfield  {author} {\bibinfo {author} {\bibfnamefont {H.~A.}\ \bibnamefont
  {{Posch}}}, \bibinfo {author} {\bibfnamefont {U.}~\bibnamefont {{Balucani}}},
  \ and\ \bibinfo {author} {\bibfnamefont {R.}~\bibnamefont {{Vallauri}}},\
  }\href@noop {} {\bibfield  {journal} {\bibinfo  {journal} {Physica A}\
  }\textbf {\bibinfo {volume} {123}},\ \bibinfo {pages} {516} (\bibinfo {year}
  {1984})}\BibitemShut {NoStop}%
\bibitem [{\citenamefont {Straub}\ \emph {et~al.}(1987)\citenamefont {Straub},
  \citenamefont {Borkovec},\ and\ \citenamefont
  {Berne}}]{Straub_Berne_JPC_1987}%
  \BibitemOpen
  \bibfield  {author} {\bibinfo {author} {\bibfnamefont {J.~E.}\ \bibnamefont
  {Straub}}, \bibinfo {author} {\bibfnamefont {M.}~\bibnamefont {Borkovec}}, \
  and\ \bibinfo {author} {\bibfnamefont {B.~J.}\ \bibnamefont {Berne}},\
  }\href@noop {} {\bibfield  {journal} {\bibinfo  {journal} {The Journal of
  Physical Chemistry}\ }\textbf {\bibinfo {volume} {91}},\ \bibinfo {pages}
  {4995} (\bibinfo {year} {1987})}\BibitemShut {NoStop}%
\bibitem [{\citenamefont {Straub}\ \emph {et~al.}(1990)\citenamefont {Straub},
  \citenamefont {Berne},\ and\ \citenamefont {Roux}}]{Straub_Berne_JCP_1990}%
  \BibitemOpen
  \bibfield  {author} {\bibinfo {author} {\bibfnamefont {J.~E.}\ \bibnamefont
  {Straub}}, \bibinfo {author} {\bibfnamefont {B.~J.}\ \bibnamefont {Berne}}, \
  and\ \bibinfo {author} {\bibfnamefont {B.}~\bibnamefont {Roux}},\ }\href@noop
  {} {\bibfield  {journal} {\bibinfo  {journal} {The Journal of Chemical
  Physics}\ }\textbf {\bibinfo {volume} {93}},\ \bibinfo {pages} {6804}
  (\bibinfo {year} {1990})}\BibitemShut {NoStop}%
\bibitem [{\citenamefont {Plotkin}\ and\ \citenamefont
  {Wolynes}(1998)}]{Plotkin_Wolynes_PRL_1998}%
  \BibitemOpen
  \bibfield  {author} {\bibinfo {author} {\bibfnamefont {S.~S.}\ \bibnamefont
  {Plotkin}}\ and\ \bibinfo {author} {\bibfnamefont {P.~G.}\ \bibnamefont
  {Wolynes}},\ }\href@noop {} {\bibfield  {journal} {\bibinfo  {journal} {Phys.
  Rev. Lett.}\ }\textbf {\bibinfo {volume} {80}},\ \bibinfo {pages} {5015}
  (\bibinfo {year} {1998})}\BibitemShut {NoStop}%
\bibitem [{\citenamefont {Luo}\ \emph {et~al.}(2006)\citenamefont {Luo},
  \citenamefont {Andricioaei}, \citenamefont {Xie},\ and\ \citenamefont
  {Karplus}}]{Luo_Xie_JPCB_2006}%
  \BibitemOpen
  \bibfield  {author} {\bibinfo {author} {\bibfnamefont {G.}~\bibnamefont
  {Luo}}, \bibinfo {author} {\bibfnamefont {I.}~\bibnamefont {Andricioaei}},
  \bibinfo {author} {\bibfnamefont {X.~S.}\ \bibnamefont {Xie}}, \ and\
  \bibinfo {author} {\bibfnamefont {M.}~\bibnamefont {Karplus}},\ }\href@noop
  {} {\bibfield  {journal} {\bibinfo  {journal} {The Journal of Physical
  Chemistry B}\ }\textbf {\bibinfo {volume} {110}},\ \bibinfo {pages} {9363}
  (\bibinfo {year} {2006})}\BibitemShut {NoStop}%
\bibitem [{\citenamefont {Best}\ and\ \citenamefont
  {Hummer}(2006)}]{Best_Hummer_PRL_2006}%
  \BibitemOpen
  \bibfield  {author} {\bibinfo {author} {\bibfnamefont {R.~B.}\ \bibnamefont
  {Best}}\ and\ \bibinfo {author} {\bibfnamefont {G.}~\bibnamefont {Hummer}},\
  }\href@noop {} {\bibfield  {journal} {\bibinfo  {journal} {Phys. Rev. Lett.}\
  }\textbf {\bibinfo {volume} {96}},\ \bibinfo {pages} {228104} (\bibinfo
  {year} {2006})}\BibitemShut {NoStop}%
\bibitem [{\citenamefont {Best}\ and\ \citenamefont
  {Hummer}(2010)}]{Best_Hummer_PNAS_2010}%
  \BibitemOpen
  \bibfield  {author} {\bibinfo {author} {\bibfnamefont {R.~B.}\ \bibnamefont
  {Best}}\ and\ \bibinfo {author} {\bibfnamefont {G.}~\bibnamefont {Hummer}},\
  }\href@noop {} {\bibfield  {journal} {\bibinfo  {journal} {Proceedings of the
  National Academy of Sciences}\ }\textbf {\bibinfo {volume} {107}},\ \bibinfo
  {pages} {1088} (\bibinfo {year} {2010})}\BibitemShut {NoStop}%
\bibitem [{\citenamefont {Hinczewski}\ \emph {et~al.}(2010)\citenamefont
  {Hinczewski}, \citenamefont {von Hansen}, \citenamefont {Dzubiella},\ and\
  \citenamefont {Netz}}]{Hinczewski_Netz_JCP_2010}%
  \BibitemOpen
  \bibfield  {author} {\bibinfo {author} {\bibfnamefont {M.}~\bibnamefont
  {Hinczewski}}, \bibinfo {author} {\bibfnamefont {Y.}~\bibnamefont {von
  Hansen}}, \bibinfo {author} {\bibfnamefont {J.}~\bibnamefont {Dzubiella}}, \
  and\ \bibinfo {author} {\bibfnamefont {R.~R.}\ \bibnamefont {Netz}},\
  }\href@noop {} {\bibfield  {journal} {\bibinfo  {journal} {The Journal of
  Chemical Physics}\ }\textbf {\bibinfo {volume} {132}},\ \bibinfo {pages}
  {245103} (\bibinfo {year} {2010})}\BibitemShut {NoStop}%
\bibitem [{\citenamefont {Satija}\ \emph {et~al.}(2017)\citenamefont {Satija},
  \citenamefont {Das},\ and\ \citenamefont
  {Makarov}}]{Satija_Makarov_JCP_2017}%
  \BibitemOpen
  \bibfield  {author} {\bibinfo {author} {\bibfnamefont {R.}~\bibnamefont
  {Satija}}, \bibinfo {author} {\bibfnamefont {A.}~\bibnamefont {Das}}, \ and\
  \bibinfo {author} {\bibfnamefont {D.~E.}\ \bibnamefont {Makarov}},\
  }\href@noop {} {\bibfield  {journal} {\bibinfo  {journal} {The Journal of
  Chemical Physics}\ }\textbf {\bibinfo {volume} {147}},\ \bibinfo {pages}
  {152707} (\bibinfo {year} {2017})}\BibitemShut {NoStop}%
\bibitem [{\citenamefont {Morrone}\ \emph {et~al.}(2012)\citenamefont
  {Morrone}, \citenamefont {Li},\ and\ \citenamefont
  {Berne}}]{Morrone_Li_JPCB_2012}%
  \BibitemOpen
  \bibfield  {author} {\bibinfo {author} {\bibfnamefont {J.~A.}\ \bibnamefont
  {Morrone}}, \bibinfo {author} {\bibfnamefont {J.}~\bibnamefont {Li}}, \ and\
  \bibinfo {author} {\bibfnamefont {B.~J.}\ \bibnamefont {Berne}},\ }\href@noop
  {} {\bibfield  {journal} {\bibinfo  {journal} {The Journal of Physical
  Chemistry B}\ }\textbf {\bibinfo {volume} {116}},\ \bibinfo {pages} {378}
  (\bibinfo {year} {2012})}\BibitemShut {NoStop}%
\bibitem [{\citenamefont {Daldrop}\ \emph {et~al.}(2017)\citenamefont
  {Daldrop}, \citenamefont {Kowalik},\ and\ \citenamefont
  {Netz}}]{daldrop2017external}%
  \BibitemOpen
  \bibfield  {author} {\bibinfo {author} {\bibfnamefont {J.~O.}\ \bibnamefont
  {Daldrop}}, \bibinfo {author} {\bibfnamefont {B.~G.}\ \bibnamefont
  {Kowalik}}, \ and\ \bibinfo {author} {\bibfnamefont {R.~R.}\ \bibnamefont
  {Netz}},\ }\href@noop {} {\bibfield  {journal} {\bibinfo  {journal} {Physical
  Review X}\ }\textbf {\bibinfo {volume} {7}},\ \bibinfo {pages} {041065}
  (\bibinfo {year} {2017})}\BibitemShut {NoStop}%
\bibitem [{\citenamefont {Deutch}\ and\ \citenamefont
  {Oppenheim}(1971)}]{Deutch_Oppenheim_JCP_1971}%
  \BibitemOpen
  \bibfield  {author} {\bibinfo {author} {\bibfnamefont {J.~M.}\ \bibnamefont
  {Deutch}}\ and\ \bibinfo {author} {\bibfnamefont {I.}~\bibnamefont
  {Oppenheim}},\ }\href@noop {} {\bibfield  {journal} {\bibinfo  {journal} {The
  Journal of Chemical Physics}\ }\textbf {\bibinfo {volume} {54}},\ \bibinfo
  {pages} {3547} (\bibinfo {year} {1971})}\BibitemShut {NoStop}%
\bibitem [{\citenamefont {Zwanzig}(1973)}]{Zwanzig73}%
  \BibitemOpen
  \bibfield  {author} {\bibinfo {author} {\bibfnamefont {R.}~\bibnamefont
  {Zwanzig}},\ }\href@noop {} {\bibfield  {journal} {\bibinfo  {journal} {J.
  Stat. Phys.}\ }\textbf {\bibinfo {volume} {9}},\ \bibinfo {pages} {215 }
  (\bibinfo {year} {1973})}\BibitemShut {NoStop}%
\bibitem [{\citenamefont {Zwanzig}(1992)}]{Zwanzig_diffusion_JPC_1992}%
  \BibitemOpen
  \bibfield  {author} {\bibinfo {author} {\bibfnamefont {R.}~\bibnamefont
  {Zwanzig}},\ }\href@noop {} {\bibfield  {journal} {\bibinfo  {journal} {The
  Journal of Physical Chemistry}\ }\textbf {\bibinfo {volume} {96}},\ \bibinfo
  {pages} {3926} (\bibinfo {year} {1992})}\BibitemShut {NoStop}%
\bibitem [{\citenamefont {Berezhkovskii}\ and\ \citenamefont
  {Szabo}(2011)}]{Berezhkovskii_Szabo_JCP_2011}%
  \BibitemOpen
  \bibfield  {author} {\bibinfo {author} {\bibfnamefont {A.}~\bibnamefont
  {Berezhkovskii}}\ and\ \bibinfo {author} {\bibfnamefont {A.}~\bibnamefont
  {Szabo}},\ }\href@noop {} {\bibfield  {journal} {\bibinfo  {journal} {The
  Journal of Chemical Physics}\ }\textbf {\bibinfo {volume} {135}},\ \bibinfo
  {pages} {074108} (\bibinfo {year} {2011})}\BibitemShut {NoStop}%
\bibitem [{\citenamefont {Glatzel}\ and\ \citenamefont
  {Schilling}(2022)}]{Glatzel_Schilling_EPL_2021}%
  \BibitemOpen
  \bibfield  {author} {\bibinfo {author} {\bibfnamefont {F.}~\bibnamefont
  {Glatzel}}\ and\ \bibinfo {author} {\bibfnamefont {T.}~\bibnamefont
  {Schilling}},\ }\href@noop {} {\bibfield  {journal} {\bibinfo  {journal}
  {Europhysics Letters}\ }\textbf {\bibinfo {volume} {136}},\ \bibinfo {pages}
  {36001} (\bibinfo {year} {2022})}\BibitemShut {NoStop}%
\bibitem [{\citenamefont {Vroylandt}(2022)}]{Vroylandt_EPL_2022}%
  \BibitemOpen
  \bibfield  {author} {\bibinfo {author} {\bibfnamefont {H.}~\bibnamefont
  {Vroylandt}},\ }\href@noop {} {\bibfield  {journal} {\bibinfo  {journal}
  {Europhysics Letters}\ }\textbf {\bibinfo {volume} {140}},\ \bibinfo {pages}
  {62003} (\bibinfo {year} {2022})}\BibitemShut {NoStop}%
\bibitem [{\citenamefont {Vroylandt}\ and\ \citenamefont
  {Monmarch{\'e}}(2022)}]{Vroylandt_Monmarche_JCP_2022}%
  \BibitemOpen
  \bibfield  {author} {\bibinfo {author} {\bibfnamefont {H.}~\bibnamefont
  {Vroylandt}}\ and\ \bibinfo {author} {\bibfnamefont {P.}~\bibnamefont
  {Monmarch{\'e}}},\ }\href@noop {} {\bibfield  {journal} {\bibinfo  {journal}
  {The Journal of Chemical Physics}\ }\textbf {\bibinfo {volume} {156}},\
  \bibinfo {pages} {244105} (\bibinfo {year} {2022})}\BibitemShut {NoStop}%
\bibitem [{\citenamefont {Ayaz}\ \emph {et~al.}(2022)\citenamefont {Ayaz},
  \citenamefont {Scalfi}, \citenamefont {Dalton},\ and\ \citenamefont
  {Netz}}]{ayaz2022generalized}%
  \BibitemOpen
  \bibfield  {author} {\bibinfo {author} {\bibfnamefont {C.}~\bibnamefont
  {Ayaz}}, \bibinfo {author} {\bibfnamefont {L.}~\bibnamefont {Scalfi}},
  \bibinfo {author} {\bibfnamefont {B.~A.}\ \bibnamefont {Dalton}}, \ and\
  \bibinfo {author} {\bibfnamefont {R.~R.}\ \bibnamefont {Netz}},\ }\href@noop
  {} {\bibfield  {journal} {\bibinfo  {journal} {Physical Review E}\ }\textbf
  {\bibinfo {volume} {105}},\ \bibinfo {pages} {054138} (\bibinfo {year}
  {2022})}\BibitemShut {NoStop}%
\bibitem [{\citenamefont {Jung}\ and\ \citenamefont
  {Jung}(2023)}]{Jung_Jung_JCP_2023}%
  \BibitemOpen
  \bibfield  {author} {\bibinfo {author} {\bibfnamefont {B.}~\bibnamefont
  {Jung}}\ and\ \bibinfo {author} {\bibfnamefont {G.}~\bibnamefont {Jung}},\
  }\href@noop {} {\bibfield  {journal} {\bibinfo  {journal} {The Journal of
  Chemical Physics}\ }\textbf {\bibinfo {volume} {159}},\ \bibinfo {pages}
  {084110} (\bibinfo {year} {2023})}\BibitemShut {NoStop}%
\bibitem [{\citenamefont {Lyu}\ and\ \citenamefont
  {Lei}(2023)}]{Lyu_Lei_PRL_2023}%
  \BibitemOpen
  \bibfield  {author} {\bibinfo {author} {\bibfnamefont {L.}~\bibnamefont
  {Lyu}}\ and\ \bibinfo {author} {\bibfnamefont {H.}~\bibnamefont {Lei}},\
  }\href@noop {} {\bibfield  {journal} {\bibinfo  {journal} {Phys. Rev. Lett.}\
  }\textbf {\bibinfo {volume} {131}},\ \bibinfo {pages} {177301} (\bibinfo
  {year} {2023})}\BibitemShut {NoStop}%
\bibitem [{\citenamefont {Satija}\ and\ \citenamefont
  {Makarov}(2019)}]{Satija_Makarov_JPCB_2019}%
  \BibitemOpen
  \bibfield  {author} {\bibinfo {author} {\bibfnamefont {R.}~\bibnamefont
  {Satija}}\ and\ \bibinfo {author} {\bibfnamefont {D.~E.}\ \bibnamefont
  {Makarov}},\ }\href@noop {} {\bibfield  {journal} {\bibinfo  {journal} {The
  Journal of Physical Chemistry B}\ }\textbf {\bibinfo {volume} {123}},\
  \bibinfo {pages} {802} (\bibinfo {year} {2019})}\BibitemShut {NoStop}%
\bibitem [{\citenamefont {Grogan}\ \emph {et~al.}(2020)\citenamefont {Grogan},
  \citenamefont {Lei}, \citenamefont {Li},\ and\ \citenamefont
  {Baker}}]{Grogan_Lei_JCP_2020}%
  \BibitemOpen
  \bibfield  {author} {\bibinfo {author} {\bibfnamefont {F.}~\bibnamefont
  {Grogan}}, \bibinfo {author} {\bibfnamefont {H.}~\bibnamefont {Lei}},
  \bibinfo {author} {\bibfnamefont {X.}~\bibnamefont {Li}}, \ and\ \bibinfo
  {author} {\bibfnamefont {N.~A.}\ \bibnamefont {Baker}},\ }\href@noop {}
  {\bibfield  {journal} {\bibinfo  {journal} {J. Comput. Phys.}\ }\textbf
  {\bibinfo {volume} {418}},\ \bibinfo {pages} {109633} (\bibinfo {year}
  {2020})}\BibitemShut {NoStop}%
\bibitem [{\citenamefont {Singh}\ \emph {et~al.}(2021)\citenamefont {Singh},
  \citenamefont {Mondal},\ and\ \citenamefont
  {Chaudhury}}]{Singh_Mondal_JPCB_2021}%
  \BibitemOpen
  \bibfield  {author} {\bibinfo {author} {\bibfnamefont {D.}~\bibnamefont
  {Singh}}, \bibinfo {author} {\bibfnamefont {K.}~\bibnamefont {Mondal}}, \
  and\ \bibinfo {author} {\bibfnamefont {S.}~\bibnamefont {Chaudhury}},\
  }\href@noop {} {\bibfield  {journal} {\bibinfo  {journal} {The Journal of
  Physical Chemistry B}\ }\textbf {\bibinfo {volume} {125}},\ \bibinfo {pages}
  {4536} (\bibinfo {year} {2021})}\BibitemShut {NoStop}%
\bibitem [{\citenamefont {Ayaz}\ \emph {et~al.}(2021)\citenamefont {Ayaz},
  \citenamefont {Tepper}, \citenamefont {Brünig}, \citenamefont {Kappler},
  \citenamefont {Daldrop},\ and\ \citenamefont {Netz}}]{Ayaz_Netz_PNAS_2021}%
  \BibitemOpen
  \bibfield  {author} {\bibinfo {author} {\bibfnamefont {C.}~\bibnamefont
  {Ayaz}}, \bibinfo {author} {\bibfnamefont {L.}~\bibnamefont {Tepper}},
  \bibinfo {author} {\bibfnamefont {F.~N.}\ \bibnamefont {Brünig}}, \bibinfo
  {author} {\bibfnamefont {J.}~\bibnamefont {Kappler}}, \bibinfo {author}
  {\bibfnamefont {J.~O.}\ \bibnamefont {Daldrop}}, \ and\ \bibinfo {author}
  {\bibfnamefont {R.~R.}\ \bibnamefont {Netz}},\ }\href@noop {} {\bibfield
  {journal} {\bibinfo  {journal} {Proceedings of the National Academy of
  Sciences}\ }\textbf {\bibinfo {volume} {118}},\ \bibinfo {pages}
  {e2023856118} (\bibinfo {year} {2021})}\BibitemShut {NoStop}%
\bibitem [{\citenamefont {Dalton}\ \emph {et~al.}(2023)\citenamefont {Dalton},
  \citenamefont {Ayaz}, \citenamefont {Kiefer}, \citenamefont {Klimek},
  \citenamefont {Tepper},\ and\ \citenamefont {Netz}}]{Dalton_Netz_PNAS_2023}%
  \BibitemOpen
  \bibfield  {author} {\bibinfo {author} {\bibfnamefont {B.~A.}\ \bibnamefont
  {Dalton}}, \bibinfo {author} {\bibfnamefont {C.}~\bibnamefont {Ayaz}},
  \bibinfo {author} {\bibfnamefont {H.}~\bibnamefont {Kiefer}}, \bibinfo
  {author} {\bibfnamefont {A.}~\bibnamefont {Klimek}}, \bibinfo {author}
  {\bibfnamefont {L.}~\bibnamefont {Tepper}}, \ and\ \bibinfo {author}
  {\bibfnamefont {R.~R.}\ \bibnamefont {Netz}},\ }\href@noop {} {\bibfield
  {journal} {\bibinfo  {journal} {Proceedings of the National Academy of
  Sciences}\ }\textbf {\bibinfo {volume} {120}},\ \bibinfo {pages}
  {e2220068120} (\bibinfo {year} {2023})}\BibitemShut {NoStop}%
\bibitem [{\citenamefont {Straub}\ \emph {et~al.}(1988)\citenamefont {Straub},
  \citenamefont {Borkovec},\ and\ \citenamefont
  {Berne}}]{Straub_Berne_JCP_1988}%
  \BibitemOpen
  \bibfield  {author} {\bibinfo {author} {\bibfnamefont {J.~E.}\ \bibnamefont
  {Straub}}, \bibinfo {author} {\bibfnamefont {M.}~\bibnamefont {Borkovec}}, \
  and\ \bibinfo {author} {\bibfnamefont {B.~J.}\ \bibnamefont {Berne}},\
  }\href@noop {} {\bibfield  {journal} {\bibinfo  {journal} {The Journal of
  Chemical Physics}\ }\textbf {\bibinfo {volume} {89}},\ \bibinfo {pages}
  {4833} (\bibinfo {year} {1988})}\BibitemShut {NoStop}%
\bibitem [{\citenamefont {Singh}\ \emph {et~al.}(1990)\citenamefont {Singh},
  \citenamefont {Krishnan},\ and\ \citenamefont
  {Robinson}}]{Singh_Krishnan_CPL_1990}%
  \BibitemOpen
  \bibfield  {author} {\bibinfo {author} {\bibfnamefont {S.}~\bibnamefont
  {Singh}}, \bibinfo {author} {\bibfnamefont {R.}~\bibnamefont {Krishnan}}, \
  and\ \bibinfo {author} {\bibfnamefont {G.}~\bibnamefont {Robinson}},\
  }\href@noop {} {\bibfield  {journal} {\bibinfo  {journal} {Chemical Physics
  Letters}\ }\textbf {\bibinfo {volume} {175}},\ \bibinfo {pages} {338}
  (\bibinfo {year} {1990})}\BibitemShut {NoStop}%
\bibitem [{\citenamefont {Carmeli}\ and\ \citenamefont
  {Nitzan}(1983)}]{Carmeli_Nitzan_CPL_1983}%
  \BibitemOpen
  \bibfield  {author} {\bibinfo {author} {\bibfnamefont {B.}~\bibnamefont
  {Carmeli}}\ and\ \bibinfo {author} {\bibfnamefont {A.}~\bibnamefont
  {Nitzan}},\ }\href@noop {} {\bibfield  {journal} {\bibinfo  {journal}
  {Chemical Physics Letters}\ }\textbf {\bibinfo {volume} {102}},\ \bibinfo
  {pages} {517} (\bibinfo {year} {1983})}\BibitemShut {NoStop}%
\bibitem [{\citenamefont {Tarjus}\ and\ \citenamefont
  {Kivelson}(1991)}]{Tarjus_Kivelson_Chemical_Physics_1991}%
  \BibitemOpen
  \bibfield  {author} {\bibinfo {author} {\bibfnamefont {G.}~\bibnamefont
  {Tarjus}}\ and\ \bibinfo {author} {\bibfnamefont {D.}~\bibnamefont
  {Kivelson}},\ }\href@noop {} {\bibfield  {journal} {\bibinfo  {journal}
  {Chemical Physics}\ }\textbf {\bibinfo {volume} {152}},\ \bibinfo {pages}
  {153} (\bibinfo {year} {1991})}\BibitemShut {NoStop}%
\bibitem [{\citenamefont {Krishnan}\ \emph {et~al.}(1992)\citenamefont
  {Krishnan}, \citenamefont {Singh},\ and\ \citenamefont
  {Robinson}}]{Krishnan_Singh_JCP_1992}%
  \BibitemOpen
  \bibfield  {author} {\bibinfo {author} {\bibfnamefont {R.}~\bibnamefont
  {Krishnan}}, \bibinfo {author} {\bibfnamefont {S.}~\bibnamefont {Singh}}, \
  and\ \bibinfo {author} {\bibfnamefont {G.~W.}\ \bibnamefont {Robinson}},\
  }\href@noop {} {\bibfield  {journal} {\bibinfo  {journal} {The Journal of
  Chemical Physics}\ }\textbf {\bibinfo {volume} {97}},\ \bibinfo {pages}
  {5516} (\bibinfo {year} {1992})}\BibitemShut {NoStop}%
\bibitem [{\citenamefont {Voth}(1992)}]{Voth_JCP_1992}%
  \BibitemOpen
  \bibfield  {author} {\bibinfo {author} {\bibfnamefont {G.~A.}\ \bibnamefont
  {Voth}},\ }\href@noop {} {\bibfield  {journal} {\bibinfo  {journal} {The
  Journal of Chemical Physics}\ }\textbf {\bibinfo {volume} {97}},\ \bibinfo
  {pages} {5908} (\bibinfo {year} {1992})}\BibitemShut {NoStop}%
\bibitem [{\citenamefont {Straus}\ \emph {et~al.}(1993)\citenamefont {Straus},
  \citenamefont {Gomez~Llorente},\ and\ \citenamefont
  {Voth}}]{Straus_Voth_JCP_1993}%
  \BibitemOpen
  \bibfield  {author} {\bibinfo {author} {\bibfnamefont {J.~B.}\ \bibnamefont
  {Straus}}, \bibinfo {author} {\bibfnamefont {J.~M.}\ \bibnamefont
  {Gomez~Llorente}}, \ and\ \bibinfo {author} {\bibfnamefont {G.~A.}\
  \bibnamefont {Voth}},\ }\href@noop {} {\bibfield  {journal} {\bibinfo
  {journal} {The Journal of Chemical Physics}\ }\textbf {\bibinfo {volume}
  {98}},\ \bibinfo {pages} {4082} (\bibinfo {year} {1993})}\BibitemShut
  {NoStop}%
\bibitem [{\citenamefont {Haynes}\ \emph {et~al.}(1993)\citenamefont {Haynes},
  \citenamefont {Voth},\ and\ \citenamefont {Pollak}}]{Haynes_Voth_CPL_1993}%
  \BibitemOpen
  \bibfield  {author} {\bibinfo {author} {\bibfnamefont {G.~R.}\ \bibnamefont
  {Haynes}}, \bibinfo {author} {\bibfnamefont {G.~A.}\ \bibnamefont {Voth}}, \
  and\ \bibinfo {author} {\bibfnamefont {E.}~\bibnamefont {Pollak}},\
  }\href@noop {} {\bibfield  {journal} {\bibinfo  {journal} {Chemical Physics
  Letters}\ }\textbf {\bibinfo {volume} {207}},\ \bibinfo {pages} {309}
  (\bibinfo {year} {1993})}\BibitemShut {NoStop}%
\bibitem [{\citenamefont {Haynes}\ \emph {et~al.}(1994)\citenamefont {Haynes},
  \citenamefont {Voth},\ and\ \citenamefont {Pollak}}]{Haynes_Voth_JCP_1994}%
  \BibitemOpen
  \bibfield  {author} {\bibinfo {author} {\bibfnamefont {G.~R.}\ \bibnamefont
  {Haynes}}, \bibinfo {author} {\bibfnamefont {G.~A.}\ \bibnamefont {Voth}}, \
  and\ \bibinfo {author} {\bibfnamefont {E.}~\bibnamefont {Pollak}},\
  }\href@noop {} {\bibfield  {journal} {\bibinfo  {journal} {The Journal of
  Chemical Physics}\ }\textbf {\bibinfo {volume} {101}},\ \bibinfo {pages}
  {7811} (\bibinfo {year} {1994})}\BibitemShut {NoStop}%
\bibitem [{\citenamefont {Cossio}\ \emph {et~al.}(2015)\citenamefont {Cossio},
  \citenamefont {Hummer},\ and\ \citenamefont
  {Szabo}}]{Cossio_Hummer_PNAS_2015}%
  \BibitemOpen
  \bibfield  {author} {\bibinfo {author} {\bibfnamefont {P.}~\bibnamefont
  {Cossio}}, \bibinfo {author} {\bibfnamefont {G.}~\bibnamefont {Hummer}}, \
  and\ \bibinfo {author} {\bibfnamefont {A.}~\bibnamefont {Szabo}},\
  }\href@noop {} {\bibfield  {journal} {\bibinfo  {journal} {Proceedings of the
  National Academy of Sciences}\ }\textbf {\bibinfo {volume} {112}},\ \bibinfo
  {pages} {14248} (\bibinfo {year} {2015})}\BibitemShut {NoStop}%
\bibitem [{\citenamefont {Torrie}\ and\ \citenamefont
  {Valleau}(1977)}]{Torrie_Valleau_Umbrella_JCP_1977}%
  \BibitemOpen
  \bibfield  {author} {\bibinfo {author} {\bibfnamefont {G.}~\bibnamefont
  {Torrie}}\ and\ \bibinfo {author} {\bibfnamefont {J.}~\bibnamefont
  {Valleau}},\ }\href@noop {} {\bibfield  {journal} {\bibinfo  {journal}
  {Journal of Computational Physics}\ }\textbf {\bibinfo {volume} {23}},\
  \bibinfo {pages} {187} (\bibinfo {year} {1977})}\BibitemShut {NoStop}%
\bibitem [{\citenamefont {Kumar}\ \emph {et~al.}(1992)\citenamefont {Kumar},
  \citenamefont {Rosenberg}, \citenamefont {Bouzida}, \citenamefont
  {Swendsen},\ and\ \citenamefont {Kollman}}]{Kumar_Kollman_JCC_1992}%
  \BibitemOpen
  \bibfield  {author} {\bibinfo {author} {\bibfnamefont {S.}~\bibnamefont
  {Kumar}}, \bibinfo {author} {\bibfnamefont {J.~M.}\ \bibnamefont
  {Rosenberg}}, \bibinfo {author} {\bibfnamefont {D.}~\bibnamefont {Bouzida}},
  \bibinfo {author} {\bibfnamefont {R.~H.}\ \bibnamefont {Swendsen}}, \ and\
  \bibinfo {author} {\bibfnamefont {P.~A.}\ \bibnamefont {Kollman}},\
  }\href@noop {} {\bibfield  {journal} {\bibinfo  {journal} {Journal of
  Computational Chemistry}\ }\textbf {\bibinfo {volume} {13}},\ \bibinfo
  {pages} {1011} (\bibinfo {year} {1992})}\BibitemShut {NoStop}%
\bibitem [{\citenamefont {Darve}\ and\ \citenamefont
  {Pohorille}(2001)}]{Darve_Pohorille_JCP_2001}%
  \BibitemOpen
  \bibfield  {author} {\bibinfo {author} {\bibfnamefont {E.}~\bibnamefont
  {Darve}}\ and\ \bibinfo {author} {\bibfnamefont {A.}~\bibnamefont
  {Pohorille}},\ }\href@noop {} {\bibfield  {journal} {\bibinfo  {journal} {The
  Journal of Chemical Physics}\ }\textbf {\bibinfo {volume} {115}},\ \bibinfo
  {pages} {9169} (\bibinfo {year} {2001})}\BibitemShut {NoStop}%
\bibitem [{\citenamefont {Laio}\ and\ \citenamefont
  {Parrinello}(2002)}]{Laio_Parrinello_PNAS_2002}%
  \BibitemOpen
  \bibfield  {author} {\bibinfo {author} {\bibfnamefont {A.}~\bibnamefont
  {Laio}}\ and\ \bibinfo {author} {\bibfnamefont {M.}~\bibnamefont
  {Parrinello}},\ }\href@noop {} {\bibfield  {journal} {\bibinfo  {journal}
  {Proceedings of the National Academy of Sciences}\ }\textbf {\bibinfo
  {volume} {99}},\ \bibinfo {pages} {12562} (\bibinfo {year}
  {2002})}\BibitemShut {NoStop}%
\bibitem [{\citenamefont {Rosso}\ \emph {et~al.}(2002)\citenamefont {Rosso},
  \citenamefont {Min\'{a}ry}, \citenamefont {Zhu},\ and\ \citenamefont
  {Tuckerman}}]{Tuckerman_adiabatic_JCP_2002}%
  \BibitemOpen
  \bibfield  {author} {\bibinfo {author} {\bibfnamefont {L.}~\bibnamefont
  {Rosso}}, \bibinfo {author} {\bibfnamefont {P.}~\bibnamefont {Min\'{a}ry}},
  \bibinfo {author} {\bibfnamefont {Z.}~\bibnamefont {Zhu}}, \ and\ \bibinfo
  {author} {\bibfnamefont {M.~E.}\ \bibnamefont {Tuckerman}},\ }\href@noop {}
  {\bibfield  {journal} {\bibinfo  {journal} {The Journal of Chemical Physics}\
  }\textbf {\bibinfo {volume} {116}},\ \bibinfo {pages} {4389} (\bibinfo {year}
  {2002})}\BibitemShut {NoStop}%
\bibitem [{\citenamefont {Maragliano}\ and\ \citenamefont
  {Vanden-Eijnden}(2006)}]{Eric_TAMD_CPL_2006}%
  \BibitemOpen
  \bibfield  {author} {\bibinfo {author} {\bibfnamefont {L.}~\bibnamefont
  {Maragliano}}\ and\ \bibinfo {author} {\bibfnamefont {E.}~\bibnamefont
  {Vanden-Eijnden}},\ }\href@noop {} {\bibfield  {journal} {\bibinfo  {journal}
  {Chemical Physics Letters}\ }\textbf {\bibinfo {volume} {426}},\ \bibinfo
  {pages} {168} (\bibinfo {year} {2006})}\BibitemShut {NoStop}%
\bibitem [{\citenamefont {Abrams}\ and\ \citenamefont
  {Tuckerman}(2008)}]{abrams2008efficient}%
  \BibitemOpen
  \bibfield  {author} {\bibinfo {author} {\bibfnamefont {J.~B.}\ \bibnamefont
  {Abrams}}\ and\ \bibinfo {author} {\bibfnamefont {M.~E.}\ \bibnamefont
  {Tuckerman}},\ }\href@noop {} {\bibfield  {journal} {\bibinfo  {journal} {The
  Journal of Physical Chemistry B}\ }\textbf {\bibinfo {volume} {112}},\
  \bibinfo {pages} {15742} (\bibinfo {year} {2008})}\BibitemShut {NoStop}%
\bibitem [{\citenamefont {Maragliano}\ and\ \citenamefont
  {Vanden-Eijnden}(2008)}]{Maragliano_Vanden_JCP_2008}%
  \BibitemOpen
  \bibfield  {author} {\bibinfo {author} {\bibfnamefont {L.}~\bibnamefont
  {Maragliano}}\ and\ \bibinfo {author} {\bibfnamefont {E.}~\bibnamefont
  {Vanden-Eijnden}},\ }\href@noop {} {\bibfield  {journal} {\bibinfo  {journal}
  {The Journal of Chemical Physics}\ }\textbf {\bibinfo {volume} {128}},\
  \bibinfo {pages} {184110} (\bibinfo {year} {2008})}\BibitemShut {NoStop}%
\bibitem [{\citenamefont {Lei}\ and\ \citenamefont
  {Li}(2021)}]{Lei_Li_JCP_2021}%
  \BibitemOpen
  \bibfield  {author} {\bibinfo {author} {\bibfnamefont {H.}~\bibnamefont
  {Lei}}\ and\ \bibinfo {author} {\bibfnamefont {X.}~\bibnamefont {Li}},\
  }\href@noop {} {\bibfield  {journal} {\bibinfo  {journal} {The Journal of
  Chemical Physics}\ }\textbf {\bibinfo {volume} {154}},\ \bibinfo {pages}
  {184108} (\bibinfo {year} {2021})}\BibitemShut {NoStop}%
\bibitem [{\citenamefont {Berkowitz}\ \emph {et~al.}(1983)\citenamefont
  {Berkowitz}, \citenamefont {Morgan},\ and\ \citenamefont
  {McCammon}}]{berkowitz1983generalized}%
  \BibitemOpen
  \bibfield  {author} {\bibinfo {author} {\bibfnamefont {M.}~\bibnamefont
  {Berkowitz}}, \bibinfo {author} {\bibfnamefont {J.}~\bibnamefont {Morgan}}, \
  and\ \bibinfo {author} {\bibfnamefont {J.~A.}\ \bibnamefont {McCammon}},\
  }\href@noop {} {\bibfield  {journal} {\bibinfo  {journal} {J. Chem. Phys.}\
  }\textbf {\bibinfo {volume} {78}},\ \bibinfo {pages} {3256} (\bibinfo {year}
  {1983})}\BibitemShut {NoStop}%
\bibitem [{\citenamefont {Ogorodnikov}\ and\ \citenamefont
  {Prigarin}(1996)}]{ogorodnikov1996numerical}%
  \BibitemOpen
  \bibfield  {author} {\bibinfo {author} {\bibfnamefont {V.~A.}\ \bibnamefont
  {Ogorodnikov}}\ and\ \bibinfo {author} {\bibfnamefont {S.~M.}\ \bibnamefont
  {Prigarin}},\ }\href@noop {} {\emph {\bibinfo {title} {Numerical Modelling of
  Random Processes and Fields: Algorithms and Applications}}}\ (\bibinfo
  {publisher} {De Gruyter},\ \bibinfo {address} {Berlin, Boston},\ \bibinfo
  {year} {1996})\BibitemShut {NoStop}%
\bibitem [{\citenamefont {Cooley}\ and\ \citenamefont
  {Tukey}(1965)}]{Coole_Tukey_FFT_1965}%
  \BibitemOpen
  \bibfield  {author} {\bibinfo {author} {\bibfnamefont {J.~W.}\ \bibnamefont
  {Cooley}}\ and\ \bibinfo {author} {\bibfnamefont {J.~W.}\ \bibnamefont
  {Tukey}},\ }\href@noop {} {\bibfield  {journal} {\bibinfo  {journal}
  {Mathematics of Computation}\ }\textbf {\bibinfo {volume} {19}},\ \bibinfo
  {pages} {297} (\bibinfo {year} {1965})}\BibitemShut {NoStop}%
\bibitem [{\citenamefont {Sch\"{a}dle}\ \emph {et~al.}(2006)\citenamefont
  {Sch\"{a}dle}, \citenamefont {L\'{o}pez-Fern\'{a}ndez},\ and\ \citenamefont
  {Lubich}}]{Achim_SIAM_2006}%
  \BibitemOpen
  \bibfield  {author} {\bibinfo {author} {\bibfnamefont {A.}~\bibnamefont
  {Sch\"{a}dle}}, \bibinfo {author} {\bibfnamefont {M.}~\bibnamefont
  {L\'{o}pez-Fern\'{a}ndez}}, \ and\ \bibinfo {author} {\bibfnamefont
  {C.}~\bibnamefont {Lubich}},\ }\href@noop {} {\bibfield  {journal} {\bibinfo
  {journal} {SIAM Journal on Scientific Computing}\ }\textbf {\bibinfo {volume}
  {28}},\ \bibinfo {pages} {421} (\bibinfo {year} {2006})}\BibitemShut
  {NoStop}%
\bibitem [{\citenamefont {Martyna}\ \emph {et~al.}(1994)\citenamefont
  {Martyna}, \citenamefont {Tobias},\ and\ \citenamefont
  {Klein}}]{Martyna_Klein_JCP_1994}%
  \BibitemOpen
  \bibfield  {author} {\bibinfo {author} {\bibfnamefont {G.~J.}\ \bibnamefont
  {Martyna}}, \bibinfo {author} {\bibfnamefont {D.~J.}\ \bibnamefont {Tobias}},
  \ and\ \bibinfo {author} {\bibfnamefont {M.~L.}\ \bibnamefont {Klein}},\
  }\href@noop {} {\bibfield  {journal} {\bibinfo  {journal} {The Journal of
  Chemical Physics}\ }\textbf {\bibinfo {volume} {101}},\ \bibinfo {pages}
  {4177} (\bibinfo {year} {1994})}\BibitemShut {NoStop}%
\bibitem [{\citenamefont {Nos\'{e}}(1984)}]{Nose_Mol_Phys_1984}%
  \BibitemOpen
  \bibfield  {author} {\bibinfo {author} {\bibfnamefont {S.}~\bibnamefont
  {Nos\'{e}}},\ }\href@noop {} {\bibfield  {journal} {\bibinfo  {journal}
  {Molecular Physics}\ }\textbf {\bibinfo {volume} {52}},\ \bibinfo {pages}
  {255} (\bibinfo {year} {1984})}\BibitemShut {NoStop}%
\bibitem [{\citenamefont {Hoover}(1985)}]{Hoover1985}%
  \BibitemOpen
  \bibfield  {author} {\bibinfo {author} {\bibfnamefont {W.~G.}\ \bibnamefont
  {Hoover}},\ }\href@noop {} {\bibfield  {journal} {\bibinfo  {journal}
  {Physical Review A}\ }\textbf {\bibinfo {volume} {31}},\ \bibinfo {pages}
  {1695} (\bibinfo {year} {1985})}\BibitemShut {NoStop}%
\bibitem [{\citenamefont {Kramers}(1940)}]{Kramers_1940}%
  \BibitemOpen
  \bibfield  {author} {\bibinfo {author} {\bibfnamefont {H.}~\bibnamefont
  {Kramers}},\ }\href@noop {} {\bibfield  {journal} {\bibinfo  {journal}
  {Physica}\ }\textbf {\bibinfo {volume} {7}},\ \bibinfo {pages} {284}
  (\bibinfo {year} {1940})}\BibitemShut {NoStop}%
\bibitem [{\citenamefont {E}\ and\ \citenamefont
  {Vanden-Eijnden}(2010)}]{E_Vanden_ARPC_2010}%
  \BibitemOpen
  \bibfield  {author} {\bibinfo {author} {\bibfnamefont {W.}~\bibnamefont {E}}\
  and\ \bibinfo {author} {\bibfnamefont {E.}~\bibnamefont {Vanden-Eijnden}},\
  }\href@noop {} {\bibfield  {journal} {\bibinfo  {journal} {Annual Review of
  Physical Chemistry}\ }\textbf {\bibinfo {volume} {61}},\ \bibinfo {pages}
  {391} (\bibinfo {year} {2010})}\BibitemShut {NoStop}%
\bibitem [{\citenamefont {Kingma}\ and\ \citenamefont
  {Ba}(2015)}]{Kingma_Ba_Adam_2015}%
  \BibitemOpen
  \bibfield  {author} {\bibinfo {author} {\bibfnamefont {D.}~\bibnamefont
  {Kingma}}\ and\ \bibinfo {author} {\bibfnamefont {J.}~\bibnamefont {Ba}},\
  }\href@noop {} {\bibfield  {journal} {\bibinfo  {journal} {International
  Conference on Learning Representations (ICLR)}\ } (\bibinfo {year}
  {2015})}\BibitemShut {NoStop}%
\bibitem [{\citenamefont {Wang}\ \emph {et~al.}(2004)\citenamefont {Wang},
  \citenamefont {Wolf}, \citenamefont {Caldwell}, \citenamefont {Kollman},\
  and\ \citenamefont {Case}}]{Wang_Wolf_Amber_JCC_2004}%
  \BibitemOpen
  \bibfield  {author} {\bibinfo {author} {\bibfnamefont {J.}~\bibnamefont
  {Wang}}, \bibinfo {author} {\bibfnamefont {R.~M.}\ \bibnamefont {Wolf}},
  \bibinfo {author} {\bibfnamefont {J.~W.}\ \bibnamefont {Caldwell}}, \bibinfo
  {author} {\bibfnamefont {P.~A.}\ \bibnamefont {Kollman}}, \ and\ \bibinfo
  {author} {\bibfnamefont {D.~A.}\ \bibnamefont {Case}},\ }\href@noop {}
  {\bibfield  {journal} {\bibinfo  {journal} {Journal of Computational
  Chemistry}\ }\textbf {\bibinfo {volume} {25}},\ \bibinfo {pages} {1157}
  (\bibinfo {year} {2004})}\BibitemShut {NoStop}%
\bibitem [{\citenamefont {Bayly}\ \emph {et~al.}(1993)\citenamefont {Bayly},
  \citenamefont {Cieplak}, \citenamefont {Cornell},\ and\ \citenamefont
  {Kollman}}]{Bayly_Cieplak_RESP_JPC_1993}%
  \BibitemOpen
  \bibfield  {author} {\bibinfo {author} {\bibfnamefont {C.~I.}\ \bibnamefont
  {Bayly}}, \bibinfo {author} {\bibfnamefont {P.}~\bibnamefont {Cieplak}},
  \bibinfo {author} {\bibfnamefont {W.}~\bibnamefont {Cornell}}, \ and\
  \bibinfo {author} {\bibfnamefont {P.~A.}\ \bibnamefont {Kollman}},\
  }\href@noop {} {\bibfield  {journal} {\bibinfo  {journal} {The Journal of
  Physical Chemistry}\ }\textbf {\bibinfo {volume} {97}},\ \bibinfo {pages}
  {10269} (\bibinfo {year} {1993})}\BibitemShut {NoStop}%
\bibitem [{\citenamefont {Jorgensen}\ \emph {et~al.}(1983)\citenamefont
  {Jorgensen}, \citenamefont {Chandrasekhar}, \citenamefont {Madura},
  \citenamefont {Impey},\ and\ \citenamefont
  {Klein}}]{Jorgensen_Chandrasekhar_TIP3P_JCP_1983}%
  \BibitemOpen
  \bibfield  {author} {\bibinfo {author} {\bibfnamefont {W.~L.}\ \bibnamefont
  {Jorgensen}}, \bibinfo {author} {\bibfnamefont {J.}~\bibnamefont
  {Chandrasekhar}}, \bibinfo {author} {\bibfnamefont {J.~D.}\ \bibnamefont
  {Madura}}, \bibinfo {author} {\bibfnamefont {R.~W.}\ \bibnamefont {Impey}}, \
  and\ \bibinfo {author} {\bibfnamefont {M.~L.}\ \bibnamefont {Klein}},\
  }\href@noop {} {\bibfield  {journal} {\bibinfo  {journal} {The Journal of
  Chemical Physics}\ }\textbf {\bibinfo {volume} {79}},\ \bibinfo {pages} {926}
  (\bibinfo {year} {1983})}\BibitemShut {NoStop}%
\bibitem [{\citenamefont {Ryckaert}\ \emph {et~al.}(1977)\citenamefont
  {Ryckaert}, \citenamefont {Ciccotti},\ and\ \citenamefont
  {Berendsen}}]{Ryckaert_Ciccotti_SHAKE_JCP_1977}%
  \BibitemOpen
  \bibfield  {author} {\bibinfo {author} {\bibfnamefont {J.-P.}\ \bibnamefont
  {Ryckaert}}, \bibinfo {author} {\bibfnamefont {G.}~\bibnamefont {Ciccotti}},
  \ and\ \bibinfo {author} {\bibfnamefont {H.~J.}\ \bibnamefont {Berendsen}},\
  }\href@noop {} {\bibfield  {journal} {\bibinfo  {journal} {Journal of
  Computational Physics}\ }\textbf {\bibinfo {volume} {23}},\ \bibinfo {pages}
  {327} (\bibinfo {year} {1977})}\BibitemShut {NoStop}%
\bibitem [{\citenamefont {Miyamoto}\ and\ \citenamefont
  {Kollman~Peter}(2004)}]{Miyamoto_Kollman_Settle_JCC_2004}%
  \BibitemOpen
  \bibfield  {author} {\bibinfo {author} {\bibfnamefont {S.}~\bibnamefont
  {Miyamoto}}\ and\ \bibinfo {author} {\bibfnamefont {A.}~\bibnamefont
  {Kollman~Peter}},\ }\href@noop {} {\bibfield  {journal} {\bibinfo  {journal}
  {Journal of Computational Chemistry}\ }\textbf {\bibinfo {volume} {13}},\
  \bibinfo {pages} {952} (\bibinfo {year} {2004})}\BibitemShut {NoStop}%
\end{thebibliography}
%

\end{document}